\documentclass[11pt]{article}
\usepackage[margin=1in]{geometry}
\usepackage{microtype}
\usepackage{hyperref}
\usepackage{amsthm}
\usepackage{amsmath}
\usepackage{amssymb}
\usepackage{graphicx}
\usepackage[shortlabels]{enumitem}
\usepackage{array}
\usepackage{tikz}
\usepackage[font=small,labelfont=bf]{caption}
\usepackage[labelformat=simple]{subcaption}

\usepackage{libertine}\usepackage[libertine]{newtxmath}
\usepackage[scaled=0.96]{zi4}

\newtheorem{theorem}{Theorem}[section]
\newtheorem{lemma}[theorem]{Lemma}
\newtheorem{observation}[theorem]{Observation}
\newtheorem{corollary}[theorem]{Corollary}

\newcommand*{\dual}[1]{\widetilde{#1}}

\theoremstyle{definition}
\newtheorem{definition}[theorem]{Definition}

\graphicspath{{figures/}{figures/gratility/}}
\newcommand\includegratility[2][]{\includegraphics[scale=0.5,#1]{#2}}

\def\makecell#1{{\def\arraystretch{1}\begin{tabular}{@{}c@{}}#1\end{tabular}}}

{\makeatletter
 \gdef\xxxmark{%
   \expandafter\ifx\csname @mpargs\endcsname\relax 
     \expandafter\ifx\csname @captype\endcsname\relax 
       \marginpar{xxx}
     \else
       xxx 
     \fi
   \else
     xxx 
   \fi}
 \gdef\xxx{\@ifnextchar[\xxx@lab\xxx@nolab}
 \long\gdef\xxx@lab[#1]#2{\textbf{[\xxxmark #2 ---{\sc #1}]}}
 \long\gdef\xxx@nolab#1{\textbf{[\xxxmark #1]}}
 \long\gdef\xxx@lab[#1]#2{}\long\gdef\xxx@nolab#1{}%
}

\usepackage{colortbl}
\definecolor{header}{rgb}{0.29,0,0.51}
\definecolor{gray}{rgb}{0.85,0.85,0.85}
\definecolor{hard}{rgb}{1,0.85,0.85}
\definecolor{open}{rgb}{1,1,0.85}
\definecolor{easy}{rgb}{0.85,0.85,1}
\def\header#1{\multicolumn{1}{|c|}{\cellcolor{header}\textcolor{white}{\textbf{#1}}}}
\def\headerr#1{\multicolumn{1}{|r|}{\cellcolor{header}\textcolor{white}{\textbf{#1}}}}
\def\zebra{\rowcolors{2}{gray!75}{white}}

\def\HARD{\cellcolor{hard}}
\def\EASY{\cellcolor{easy}}

\newcommand*{\defn}[1]{\textbf{\textit{\boldmath{#1}}}}

\begin{document}

\title{ASP-Completeness of Hamiltonicity in Grid Graphs, \\
  with Applications to Loop Puzzles}

\author{%
  MIT Hardness Group%
    \thanks{Artificial first author to highlight that the other authors (in
      alphabetical order) worked as an equal group. Please include all
      authors (including this one) in your bibliography, and refer to the
      authors as “MIT Hardness Group” (without “et al.”).}
\and
  Josh Brunner%
    \thanks{MIT Computer Science and Artificial Intelligence Laboratory,
      32 Vassar St., Cambridge, MA 02139, USA, \protect\url{{brunnerj,lkdc,edemaine,diomidova,della,tockman}@mit.edu}}
\and
  Lily Chung\footnotemark[2]
\and
  Erik D. Demaine\footnotemark[2]
\and
  Jenny Diomidova\footnotemark[2]
\and
  Della Hendrickson\footnotemark[2]
\and
  Andy Tockman\footnotemark[2]
}

\date{}

\maketitle

\begin{abstract}
  We prove that Hamiltonicity in maximum-degree-3 grid graphs
  (directed or undirected) is ASP-complete,
  i.e., it has a parsimonious reduction from every NP search problem
  (including a polynomial-time bijection between solutions).
  As a consequence, given $k$ Hamiltonian cycles,
  it is NP-complete to find another; and
  counting Hamiltonian cycles is \#P-complete.
  If we require the grid graph's vertices to form a full $m \times n$ rectangle,
  then we show that Hamiltonicity remains ASP-complete if the
  edges are directed or if we allow removing some edges
  (whereas including all undirected edges is known to be easy).
  These results enable us to develop a stronger ``T-metacell'' framework for
  proving ASP-completeness of rectangular puzzles, which requires building
  just a single gadget representing a degree-3 grid-graph vertex.
  We apply this general theory to prove ASP-completeness of
  38 pencil-and-paper puzzles where the goal is to draw a loop
  subject to given constraints:
  Slalom, Onsen-meguri, Mejilink, Detour, Tapa-Like Loop, Kouchoku, Icelom;
  Masyu, Yajilin, Nagareru, Castle Wall, Moon or Sun, Country Road, Geradeweg,
  Maxi Loop, Mid-loop, Balance Loop, Simple Loop, Haisu,
  Reflect Link, Linesweeper;
  Vertex/Touch Slitherlink, Dotchi-Loop, Ovotovata, Building Walk, Rail Pool,
  Disorderly Loop, Ant Mill, Koburin, Mukkonn Enn, Rassi Silai, (Crossing) Ichimaga, Tapa,
  Canal View, Aqre, and Paintarea.
  The last 14 of these puzzles were not even known to be NP-hard.
  Along the way, we prove ASP-completeness of some simple forms of
  Tree-Residue Vertex-Breaking (TRVB), including planar multigraphs with
  degree-6 breakable vertices, or with degree-4 breakable and
  degree-1 unbreakable vertices.
\end{abstract}

\section{Introduction}

Hamiltonicity is one of the core NP-complete problems,
used as the basis for countless NP-hardness reductions.
It accounts for two of Karp's 21 NP-complete problems
\cite{Karp-1972}: directed and undirected Hamiltonian cycle.
It has been shown to remain NP-complete for many restricted graph classes:
undirected maximum-degree-3 graphs \cite{Garey-Johnson-Stockmeyer-1974},
undirected bipartite graphs \cite{Krishnamoorthy-1975},
undirected 3-connected 3-regular bipartite graphs
\cite{Akiyama-Nishizeki-Saito-1980},
undirected 2-connected 3-regular bipartite planar graphs
\cite{Akiyama-Nishizeki-Saito-1980},
undirected 3-connected 3-regular planar graphs of minimum face degree $5$
\cite{Garey-Johnson-Tarjan-1976},
directed planar graphs with indegree and outdegree at most $2$ and total
degree at most $3$ \cite{plesnik},
and so on.

One of the most useful special cases of Hamiltonicity is (square)
\defn{grid graphs}: graphs whose vertices are a subset of the 2D integer
lattice, with an edge connecting two vertices exactly when they have
distance~$1$.
Itai, Papadimitriou, and Szwarcfiter \cite{undir_rect}
proved that Hamiltonicity is NP-complete in grid graphs.
Papadimitriou and Vazirani \cite{papadimitriou}
improved this result by proving Hamiltonicity NP-complete
in grid graphs of maximum degree~$3$.
Together, these results strengthen most of the special graph classes
mentioned above (as grid graphs are necessarily planar and bipartite),
with a stronger geometric guarantee.
Other papers extend these results to other 2D grids
\cite{Arkin,Rudoy-grid,Hou-Lynch-2018}.
Hamiltonicity in grid graphs is the foundation for NP-hardness proofs of
countless games and puzzles, from video games
\cite{Forisek-2010,Portal_FUN2018,Witness_TCS}
to pencil-and-paper puzzles
\cite{Yato-2000,Andersson-2009},
as well as practical problems such as lawn mowing and milling
\cite{Arkin-Fekete-Mitchell-2000,Milling}.

But what about \defn{parsimonious} reductions
that preserve the number of solutions?
A particularly strong form of this notion is ASP-completeness:
an NP search problem $P$ is \defn{ASP-complete} \cite{Yato-Seta-2003}
if there is a polynomial-time reduction
from every NP search problem $Q$ to $P$ along with
a polynomial-time bijection converting every solution of $P$
to a unique solution of $Q$ and vice versa.
If $P$ is ASP-complete, then the decision version of $P$ is NP-complete,
counting solutions to $P$ is \#P-complete,
and the \defn{$k$-ASP} $P$ problem --- given an instance of $P$ and
$k$ solutions, find another solution --- is NP-complete for any $k \geq 0$
\cite{Yato-Seta-2003}.

Only a few versions of Hamiltonicity are known to be ASP-complete,
or weaker, \#P-complete.
Li\'skiewicz, Ogihara, and Toda \cite{Liskiewicz2003}
proved \#P-completeness of Hamiltonicity in
undirected 3-regular planar graphs (based on \cite{Garey-Johnson-Tarjan-1976}).
Seta \cite{Seta02thecomplexities} proved ASP-completeness of
Hamiltonicity in undirected maximum-degree-3 planar graphs
(based on \cite{plesnik}).
Bosboom et al.~\cite{LessThanEdgeMatching_JIP} proved ASP-completeness of
Hamiltonicity in \emph{directed} 3-regular (indegree 2 and outdegree 1
or vice versa) planar graphs (based on \cite{plesnik}).
But what about grid graphs?

\subsection{Our Results}

In this paper, we prove that Hamiltonicity in maximum-degree-3 grid graphs
is ASP-complete.  Thus this popular problem can serve as a foundation
for ASP-completeness proofs as well.
The same result holds for Hamiltonicity in \emph{directed} maximum-degree-3
grid graphs, where each edge has a specified direction.
As mentioned above, grid graphs are bipartite and planar,
so these results roughly strengthen the ASP-completeness results mentioned
above, except that
we can guarantee ``maximum-degree-3'' but not ``3-regular''.
(No grid graphs are 3-regular; consider the top-left corner.
Furthermore, undirected 3-regular graphs have an even number of Hamiltonian cycles
by Smith's Theorem \cite{Tutte-1946}, so we cannot hope for ASP-completeness
in this case: the 1-ASP decision problem is trivial,
while the 1-ASP construction problem is in PPA \cite{Papadimitriou-1994}.)

The basis for this result is another form of Hamiltonicity called
\defn{Tree-Residue Vertex-Breaking (TRVB)} \cite{trvb},
previously used to analyze Hamiltonicity in grid graphs \cite{Rudoy-grid}.
In TRVB, we are given a graph where some vertices are breakable,
and the goal is to \defn{break} a subset of the breakable vertices ---
replacing each broken degree-$k$ vertex with $k$ degree-$1$ vertices ---
to make the graph into a tree.
This problem has a known characterization of what degrees of breakable or
unbreakable vertices make the problem polynomial vs.\ NP-complete \cite{trvb}.
We prove that several forms of TRVB are in fact ASP-complete,
including planar multigraphs with degree-6 breakable vertices,
and planar multigraphs with degree-4 breakable and
degree-1 unbreakable vertices.

We also study even more geometric forms of grid-graph Hamiltonicity.
Suppose instead of allowing an arbitrary set of vertices on the square grid,
we require the vertex set to be an entire $m \times n$ rectangle of
integer points.
Such graphs are known as \defn{rectangular grid graphs} \cite{undir_rect}.
In this case, undirected Hamiltonicity is known to be easy \cite{undir_rect}.
But we show that \emph{directed} Hamiltonicity in rectangular grid graphs
is ASP-complete.
Alternatively, if the graph is undirected but we allow removing some edges
(but not vertices)
from the rectangular grid --- a spanning subgraph of a rectangular grid graph
--- then Hamiltonicity is also ASP-complete.
Table~\ref{tbl:summary} summarizes these results.

\begin{table}
  \centering
  \renewcommand{\arraystretch}{1.5}
  \def\scale{0.85}
  \begin{tabular}{|r|c|c|c|}
    \cline{2-4}
    \multicolumn{1}{r}{}
    & \header{Rectangular} &
    \header{\makecell{Max-degree-3 spanning \\ subgraph of rectangular}} &
    \header{Max-degree-3} \\ \hline
    \headerr{Undirected} &
    \EASY P \cite{undir_rect} &
    \HARD ASP-complete [\S\ref{sec:spanning}] &
    \HARD ASP-complete [\S\ref{sec:induced}]
    \\
    \cellcolor{header}
    & \EASY \includegraphics[scale=\scale]{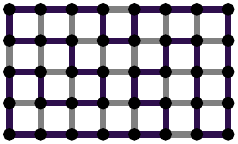}
    & \HARD \includegraphics[scale=\scale]{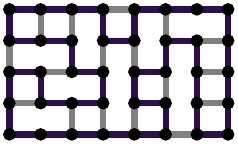}
    & \HARD \includegraphics[scale=\scale]{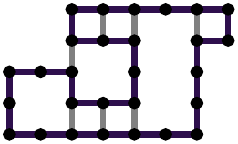}
    \\ \hline
    \headerr{Directed} &
    \HARD ASP-complete [\S\ref{sec:rectangular}] &
    \HARD ASP-complete [\S\ref{sec:spanning}] &
    \HARD ASP-complete [\S\ref{sec:induced}]
    \\
    \cellcolor{header}
    & \HARD \includegraphics[scale=\scale]{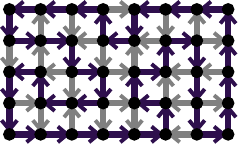}
    & \HARD \includegraphics[scale=\scale]{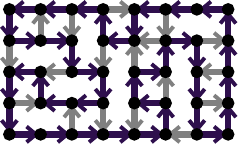}
    & \HARD \includegraphics[scale=\scale]{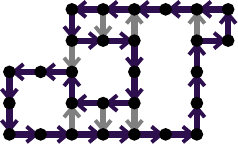}
    \\ \hline
  \end{tabular}
  \caption{
    Complexity of Hamiltonicity in various types of grid graphs.
    Each cell shows an example of a Hamiltonian graph of the specified type,
    with a darkened Hamiltonian cycle.
    The first and third column concern true grid graphs,
    where there is an edge between each pair of vertices at distance~$1$.
    In the first and second columns, the vertices form exactly
    an $m \times n$ rectangle, whereas the third column allows an
    induced subgraph of a rectangular grid graph.
    The middle column concerns graphs constructed from a rectangular grid graph
    by removing some edges (but no vertices) so that each
    vertex has degree at most~$3$.
    The second and third columns have maximum degree~$3$.
  }
  \label{tbl:summary}
\end{table}

Rectangular grid graphs are useful because many (if not most)
pencil-and-paper puzzles take place on a full rectangular grid.
In particular, the \defn{T-metacell framework} of Tang \cite{tang}
shows how NP-hardness for a pencil-and-paper puzzle
often follows from building a single gadget, essentially representing a
degree-3 vertex that must be visited at least once.
In Section~\ref{sec:T-metacells}, we extend this framework to
prove ASP-completeness as well.
We also extend the framework to allow for T-metacells where
some exits are directed (usable in only one direction)
and up to one exit is forced (must be used).
In some cases, we need to build more than one T-metacell
to handle different orientations of directions and/or forced edges.

Finally, in Section~\ref{sec:Applications}, we apply this framework
to prove ASP-completeness of 38 pencil-and-paper puzzles,
listed in Table~\ref{tbl:puzzles}.
Five of these results use the same reduction from \cite{tang},
while the remainder involve creating new T-metacell gadget(s).
For fourteen of the analyzed puzzles, even our NP-hardness result is new.

\begin{table}
  \centering
  \renewcommand{\arraystretch}{1.5}
  \zebra
  \begin{tabular}{|m{3.4in}r|ccc|}
    \hline
    \header{Games} & \header{\#} & \header{\makecell{New ASP-\\Hardness}} & \header{\makecell{New\\Reduction}} & \header{\makecell{New NP-\\Hardness}}
    \\\hline
    Slalom/Suraromu \cite{Kanehiro,tang}, Onsen-meguri \cite{tang}, Mejilink \cite{tang}, Detour \cite{tang-old,tang}, Tapa-Like Loop \cite{tang}, Kouchoku \cite{tang}, Icelom \cite{tang}
    & 7 & yes & no & no
    \\\hline
    Masyu \cite{Friedman-masyu,tang}, Yajilin \cite{Ishibashi,tang}, Nagareru \cite{Iwamoto-moon-nagare-nuri,tang}, Castle Wall \cite{tang}, Moon or Sun \cite{Iwamoto-moon-nagare-nuri,tang}, Country Road \cite{Ishibashi,tang}, Geradeweg \cite{tang},
    Maxi Loop \cite{tang}, Mid-loop \cite{tang}, Balance Loop \cite{tang}, Simple Loop \cite{undir_rect,tang}, Haisu \cite{tang-old,tang},
    Reflect Link \cite{tang}, Linesweeper \cite{maarse}
    & 14 & yes & yes & no
    \\\hline
    Vertex/Touch Slitherlink, Dotchi-Loop, Ovotovata, Building Walk, Rail Pool,
    Disorderly Loop, Ant Mill, Koburin, Mukkonn Enn, Rassi Silai, (Crossing) Ichimaga, Tapa,
    Canal View, Aqre, Paintarea
    & 17 & yes & yes & yes
    \\\hline
  \end{tabular}
  \caption{Our results on pencil-and-paper puzzles.
    All ASP-completeness results are new;
    some are via an existing reduction \cite{tang}
    and some are via a new reduction;
    and some puzzles were not even known to be NP-hard.
    (Puzzles known to be NP-hard have corresponding citations.)}
  \label{tbl:puzzles}
\end{table}

\section{Connections Between Problems}

We collect together some useful equivalences between problems on plane graphs, which are variously present in the literature \cite{barnette, trvb}.

\begin{definition}[\cite{trvb}]
  The \defn{Tree-Residue Vertex-Breaking} (TRVB) problem takes place on an undirected multigraph
  with vertices marked as either `breakable' or `unbreakable'.
  The goal is to \emph{break} a subset $S$ of the breakable vertices to leave a tree ---
  to break a vertex of degree $d$,
  replace it with $d$ new leaves attached to its incident edges.
  In other words, the graph obtained from $G$ by subdividing every edge and deleting the vertices in $S$ must be a tree.
\end{definition}

\begin{definition}[\cite{BENT1987, barnette}]
  Given a plane multigraph, a \defn{kiki Euler tour} is a cycle which traverses every edge exactly once, such that any time the cycle enters a vertex via an edge $e$, it leaves by an edge adjacent to $e$ in the cyclic order.%
  \footnote{This notion is one of two definitions of ``nonintersecting'' or ``noncrossing Euler tour''. We avoid this term to avoid confusion with the other definition, where an Euler tour is has a \defn{crossing} if there are four edges $e, e', f, f'$ adjacent to a single vertex so that $e'$ follows $e$ and $f'$ follows $f$ in the tour, and $\{e, e'\}$ alternates with $\{f, f'\}$ in the cyclic order \cite{TsaiWest3Color}.  Noncrossing Euler tours in this sense always exist, whereas kiki is a stricter condition.}
\end{definition}

The following is a well-known result with a long history; see \cite{TsaiWest3Color}.
\begin{theorem}
  \label{thm:tri3}
  Every Eulerian plane graph where every face is a triangle,
  except possibly the exterior face (a ``near-triangulation''),
  has a proper vertex 3-coloring.
\end{theorem}
Let $G$ be a connected 3-regular bipartite plane multigraph, and let $\widetilde{G}$ be its plane dual.
By Theorem~\ref{thm:tri3}, $\widetilde{G}$ is 3-colorable; equivalently it is possible to 3-color the faces of $G$ so that adjacent faces have different colors,
where faces are regarded as adjacent if they share an edge.
Note that in such a 3-coloring,
the three faces around a single vertex contain each color exactly once.

Let us fix such a coloring using the colors \{white, blue, yellow\} such that the exterior face is colored white.
Define the following graphs:
\begin{itemize}
\item $G_1$ is the directed plane multigraph obtained from $G$ by orienting every blue face clockwise and every white face counterclockwise.  This fully determines the orientation.
\item $G_2$ is the plane multigraph obtained from $G$ by contracting every yellow face to a single vertex.
\item $G_3$ is the subgraph of $\widetilde{G}$ induced by the non-white vertices.
\end{itemize}

\begin{figure}
  \centering
  \begin{subfigure}[t]{0.3\textwidth}
    \centering
    \includegraphics[width=\textwidth]{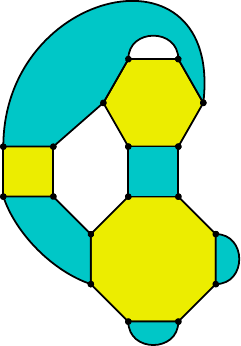}
    \caption{Face 3-coloring of $G$.}
  \end{subfigure}\hfil\hfil
  \begin{subfigure}[t]{0.3\textwidth}
    \centering
    \includegraphics[width=\textwidth]{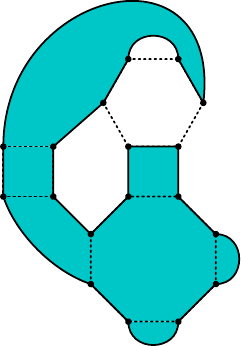}
    \caption{Assignment of colors with blue and white connected.}
  \end{subfigure}\hfil\hfil
  \begin{subfigure}[t]{0.3\textwidth}
    \centering
    \includegraphics[width=\textwidth]{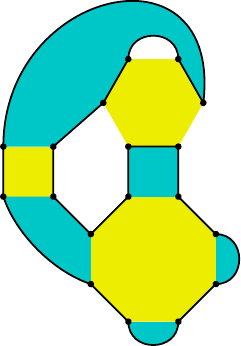}
    \caption{Cycle containing blue faces and not white faces.}
  \end{subfigure}
  \\[2em]
  \begin{subfigure}[t]{0.3\textwidth}
    \centering
    \includegraphics[width=\textwidth]{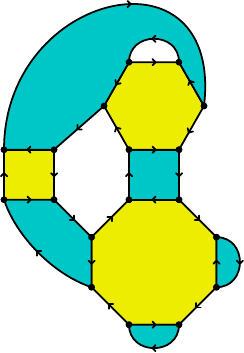}
    \caption{Directed graph $G_1$.}
  \end{subfigure}\hfil\hfil
  \begin{subfigure}[t]{0.3\textwidth}
    \centering
    \includegraphics[width=\textwidth]{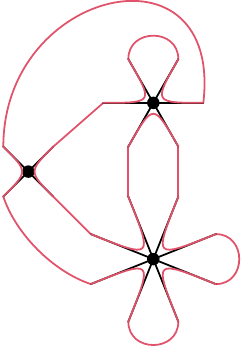}
    \caption{Kiki Euler tour of $G_2$.}
  \end{subfigure}\hfil\hfil
  \begin{subfigure}[t]{0.3\textwidth}
    \centering
    \includegraphics[width=\textwidth]{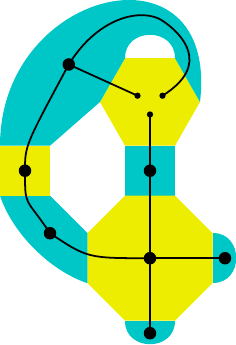}
    \caption{Tree-Residue Vertex-Breaking of $G_3$.}
  \end{subfigure}
  \captionsetup{justification=centering}
  \caption{Illustration of Lemma~\ref{lem:connections}.}
  \label{fig:connections}
\end{figure}

\begin{lemma} \label{lem:connections}
  There are bijections between the following sets:
  \begin{enumerate}[(i)]
  \item Assignments of colors \{white, blue\} to each yellow vertex of $\dual{G}$ such that the white induced subgraph is connected and the blue induced subgraph is also connected.
  \item Hamiltonian cycles of $G$ which contain all blue faces and no white faces.
  \item Hamiltonian cycles of $G$ which use every edge separating white faces from blue faces.
  \item Directed Hamiltonian cycles of $G_1$.
  \item Kiki Euler tours of $G_2$.
  \item Tree-Residue Vertex-Breakings of $G_3$, where yellow vertices are breakable and blue vertices are unbreakable.
  \end{enumerate}
\end{lemma}
\begin{proof}
  Refer to Figure~\ref{fig:connections}.
  We give explicit transformations between the sets; it can be checked that these transformations invert each other as needed.
  Figure~\ref{fig:connectgraph} summarizes the transformations we describe,
  which form a strongly connected graph.
  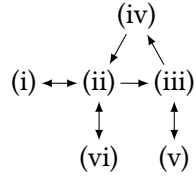
\begin{figure}
    \centering
    \begin{tikzpicture}[shorten >= -2pt, shorten <= -2pt, >=latex]

      \node[align=center] (i) at (-1,0) {(i)};
      \node[align=center] (ii) at (0,0) {(ii)};
      \node[align=center] (iii) at (1,0) {(iii)};
      \node[align=center] (iv) at (.5,.9) {(iv)};
      \node[align=center] (v) at (1,-1) {(v)};
      \node[align=center] (vi) at (0,-1) {(vi)};

      \draw[<->] (i) -- (ii);
      \draw[<->] (vi) -- (ii);
      \draw[<->] (v) -- (iii);
      \draw[->] (ii) -- (iii);
      \draw[->] (iii) -- (iv);
      \draw[->] (iv) -- (ii);

    \end{tikzpicture}
    \caption{The bijections we define for Lemma~\ref{lem:connections}.}
    \label{fig:connectgraph}
  \end{figure}

  \begin{description}
    \item[(i) $\to$ (ii):]
      Consider an assignment of colors to faces of $G$.
      For each vertex, two of the faces around it are one color
      and the third is the other color,
      so exactly two edges incident to it separate blue from white.
      The set of all edges separating blue from white thus
      forms a collection of cycles visiting each vertex once.

      We claim that this is actually a single cycle.
      If it were multiple cycles, they would
      divide the plane into more than two regions.
      Two of those regions must be the same color (blue or white),
      violating the assumption that each color is connected.

      So we have a Hamiltonian cycle separating blue from white,
      and since the exterior face is white,
      it contains all blue faces and no white faces of $G$.

    \item[(ii) $\to$ (i):]
      Given a cycle, assign blue to exactly the faces it contains.
      Since the cycle is Hamiltonian, it does not intersect itself,
      so the blue faces are connected and the white faces are connected.

    \item[(ii) $\to$ (iii):] If $C$ contains all blue faces and no white faces, then it must use every edge separating white from blue.

    \item[(iii) $\to$ (iv):] If $C$ is a cycle on $G_1$ which uses every edge separating white from blue,
      then at each individual vertex it is impossible for $C$ to reverse directions; thus it is always consistent with the orientations, so it is a directed Hamiltonian cycle.

    \item[(iv) $\to$ (ii):]
      Suppose $C$ is a directed Hamiltonian cycle of $G_1$.
      Since $C$ visits every vertex, it contains at least one edge of every face.
      Because $C$ contains an edge of the exterior face its orientation must be consistently clockwise.
      Therefore $C$ it encounters every blue face on its right side and every white face on the left, meaning it contains every blue face and does not contain any white faces.

    \item[(iii) $\to$ (v):]
      The edges separating white and blue faces are exactly the edges of $G_3$ remaining after contracting the yellow faces.
      Let $C$ be a Hamiltonian cycle of $G$ containing every white-blue edge, and let $C'$ be the Euler tour of $G_3$ obtained from $C$ by the contraction.
      It must be the case that $C$ contains exactly half of the edges incident to each yellow face, each of which connects two adjacent white-blue edges; so $C'$ is kiki.

    \item[(v) $\to$ (iii):]
      Suppose $C'$ is a kiki Euler tour of $G_3$.
      Let $C$ be the set of edges of $G$ consisting of all white-blue edges, together with those that connect consecutive edges in $C'$;
      then $C$ is a Hamiltonian cycle of $G$ containing every white-blue edge.

    \item[(ii) $\to$ (vi):]
      Note that $G_3$ does not have any edges between two
      breakable vertices,
      so breaking a vertex is equivalent to
      removing it and all incident edges.
      Thus TRVB becomes
      ``find an induced subgraph of $G_3$ containing all unbreakable vertices
      which is a tree''.

      Given a cycle $C$,
      break all yellow vertices which are outside $C$,
      or equivalently take the induced subgraph on vertices inside $C$.
      This subgraph is clearly connected.
      If it has a cycle, there is a face of $\dual G$ inside that cycle,
      which corresponds to a vertex $v$ of $G$.
      Then $v$ is strictly inside $C$.
      But $v$ must touch a white face,
      contradicting the fact that all white faces are outside $C$.
      Hence the induced subgraph on vertices inside $C$ is a tree.

    \item[(vi) $\to$ (ii):]
      Take $C$ to be the boundary
      of the tree containing blue faces
      and nonbroken yellow faces.
      Then $C$ is a cycle because it bounds a tree,
      its interior contains all blue faces (which cannot be broken)
      and no white faces (which are not present in $G_3$.
      Finally, $C$ is Hamiltonian because
      every vertex is incident to an edge separating blue from white,
      which must be in $C$.
      \qedhere

  \end{description}
\end{proof}
Furthermore, given any of the graphs $G_i$, equivalents to the others can be obtained by analogous transformations.  So these various problems can be regarded as equivalent.

An important special case of TRVB is when every breakable vertex has degree at most 3.  For planar graphs this condition is equivalent to requiring that every yellow face of the graph $G$ in the preceding discussion is a digon or triangle; it is also equivalent to kiki Euler tour with vertices of degree at most 6.
In this case, the problem can be solved in polynomial time by reducing it to a matroid parity problem.\cite{barnette}\cite{trvb}
In the next section we will discuss breakable vertices with higher degrees, with which the problem turns out to be ASP-complete.

The above characterization in Lemma~\ref{lem:connections} is general enough to usefully characterize all directed Hamiltonian max-degree-3 plane graphs. In particular, for any directed max-degree-3 plane graph $G$, it is possible to construct in polynomial time a spanning subgraph $H$ which contains all of the Hamiltonian cycles of $G$, and is essentially of the form of $G_1$ in Lemma~\ref{lem:connections}. We first give two useful facts about directed planar Hamiltonian cycles before showing how to construct $H$.

\begin{lemma}
\label{lem:cycle consistent}
Let $G$ be a directed plane multigraph, $C$ be a directed cycle of $G$, and $F$ be a face of $G$. Then every edge of $C$ that borders $F$ must have the same direction (clockwise or counterclockwise) around $F$.
\end{lemma}
\begin{proof}
A cycle in a plane graph splits the plane into two regions: inside the plane and outside the plane. Every face must lie either entirely inside or entirely outside the cycle. Any time a face touches the cycle, it must have the same orientation as the cycle: if the cycle is clockwise, then wherever it touches a face inside the cycle that edge must be oriented clockwise with respect to the face, and wherever it touches a face outside the cycle that edge must be oriented counterclockwise with respect to the face. Since every face lies entirely inside or entirely outside the cycle, all of the orientations of edges touching the face that are part of the cycle must be consistent.
\end{proof}
\begin{observation}
\label{obs:forced edges}
Let $G$ be a Hamiltonian directed plane multigraph, and let $C$ be a Hamiltonian cycle. Then any edge which is the only incoming or only outgoing edge from a vertex must be present in $C$.
\end{observation} 

\begin{lemma}
  Let $G$ be a directed max-degree-3 plane multigraph.
  Then there exists a polynomial-time algorithm
  which either reports that $G$ has no Hamiltonian cycles,
  or computes a directed spanning subgraph $H$ of $G$
  so that the faces of $H$ can be 3-colored with \{blue, white, yellow\}
  so that every blue face is oriented clockwise and every white face counterclockwise,
  such that every Hamiltonian cycle of $G$ is contained in $H$.
\end{lemma}
\begin{proof}
We describe a polynomial-time algorithm to compute $H$ from $G$.

Call an edge \emph{forced} if it must be in every Hamiltonian cycle by Observation~\ref{obs:forced edges}.  If a vertex has two forced edges, the third edge can never be taken. 

We now give 2 rules to remove edges from $G$. To get $H$, repeatedly apply these rules until they do not remove any edges; the resulting graph is $H$.  If at any point a vertex has indegree 0 or outdegree 0, then terminate and report that $G$ has no Hamiltonian cycles.

\begin{enumerate}
  \item For every vertex with two forced edges, delete any other edge incident to that vertex.
  \item For every face, delete all edges whose orientation is not consistent with every forced edge on that face.
\end{enumerate}

The first rule clearly does not remove any Hamiltonian cycles. By Lemma~\ref{lem:cycle consistent}, the second rule will also never delete an edge that is part of any Hamiltonian cycle. Thus, $H$ and $G$ have the same Hamiltonian cycles.

All that remains is to show that the faces of $H$ can be colored appropriately. We will first show that every face is one of three types:
\begin{enumerate}
  \item Every edge is oriented clockwise (blue faces)
  \item Every edge is oriented counterclockwise (white faces)
  \item The edges alternate orientation around the face (yellow faces)
\end{enumerate}

Consider a face $F$. Suppose $F$ has at least one forced edge. Then because the 2nd rule does not remove any edges from $H$, so every edge on $F$ must have the same orientation.

Now suppose $F$ has no forced edges. Consider any vertex incident to $F$. Then since neither of the edges of $F$ incident to that vertex are forced, they must either both point into or both point out of that vertex. This means that these edges must be oriented opposite ways (i.e. one clockwise and one counterclockwise) around $F$. Since this is true at every vertex of $F$, the edges alternate around $F$.

Finally, we need to show that if two faces share an edge, they must be of different types. It's never possible for two blue faces to share an edge, because if an edge is clockwise according to one of the faces, it must be counterclockwise according to the other. A similar argument applies for white faces. For yellow faces, consider a vertex $v$ incident to the shared edge is incident to.  Since edges alternate orientation around yellow faces, it follows that $v$ has either indegree 0 or outdegree 0, which is impossible.
\end{proof}

\section{ASP-Completeness of Tree-Residue Vertex-Breaking}

Demaine and Rudoy \cite{trvb} prove several NP-hardness results for TRVB
using reductions from finding Hamiltonian cycles on a max-degree-3 planar directed graph.
At the time, this Hamiltonian cycle problem was not known ASP-complete,
so they did not consider ASP-completeness.

More recently, Bosboom et al.\ \cite{LessThanEdgeMatching_JIP} showed that finding Hamiltonian cycles
on a directed max-degree-3 planar graph is ASP-complete,
using a reduction from positive 1-in-3SAT.

Several of the reductions used by Demaine and Rudoy \cite{trvb}
are easily verified to be parsimonious,
proving ASP-completeness.
We are specifically interested in the results of Section~4, on planar $(\{k\},\{4\})$-TRVB.

They first reduce finding Hamiltonian cycles on a max-degree-3 planar directed graph to
finding Hamiltonian cycles on a planar graph where all vertices have indegree and outdegree 2
and vertices have their two in-edges and their two out-edges adjacent
in the planar embedding.
This last condition is called \defn{non-alternating},
because vertices are not allowed to alternate in-edges and out-edges.
The reduction is by contracting forced edges,
and is straightforwardly parsimonious.

\begin{theorem}\label{thm:in2out2}
  Finding Hamiltonian cycles on non-alternating indegree-2 outdegree-2 planar graphs is ASP-complete.
\end{theorem}

Next, Demaine and Rudoy reduce this problem to a version of Tree-Residue Vertex-Breaking.
Specifically, Demaine and Rudoy \cite{trvb} prove NP-hardness of TRVB
on a planar graph where each unbreakable vertex has degree $4$
and each breakable vertex has degree $k$, for any constant $k\geq 4$.
This is \defn{planar $(\{k\},\{4\})$-Tree-Residue Vertex-Breaking}.
This reduction is a bit more complicated
(see Section~4.2 and in particular Figures~11 through 13 of \cite{trvb})
but it is again parsimonious;
indeed, \cite[Lemmas~4.14 and 4.15]{trvb} show that there is a bijection
between Hamiltonian cycles in the input problem
and solutions to the TRVB instance.

\begin{theorem}\label{thm:trvbk4}
  Planar $(\{k\},\{4\})$-TRVB is ASP-complete, for each $k\geq 4$.
\end{theorem}

To further simplify our reductions, we will use a slightly simpler version of TRVB:
degree-4 breakable vertices and degree-1 unbreakable vertices.

\begin{theorem}\label{thm:trvb41}
  Planar $(\{4\},\{1\})$-TRVB is ASP-complete.
\end{theorem}

\begin{proof}
  It suffices to parsimoniously simulate
  a degree-4 unbreakable vertex.
  Such a simulation is shown in Figure~\ref{fig:trvb1-4}.
  No vertex in the simulation can be broken in a solution to TRVB.
  \begin{figure}
    \centering
    \includegraphics[width=.4\linewidth]{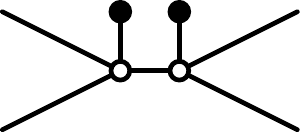}
    \caption{
      Simulating a degree-4 unbreakable vertex
      using degree-4 breakable vertices (white)
      and degree-1 unbreakable vertices (black).
    }
    \label{fig:trvb1-4}
  \end{figure}
\end{proof}

\begin{theorem}\label{thm:trvb6}
  Planar $(\{6\},\emptyset)$-TRVB is ASP-complete.
\end{theorem}

\begin{proof}
  It again suffices to simulate a degree-4 breakable vertex.
  Such a simulation is shown in Figure~\ref{fig:trvb6}.
  If the top vertex is not broken,
  both others must be broken, disconnecting the middle edge.
  So the top vertex must be broken,
  and then the other two vertices must not be.
  \begin{figure}
    \centering
    \includegraphics[width=.4\linewidth]{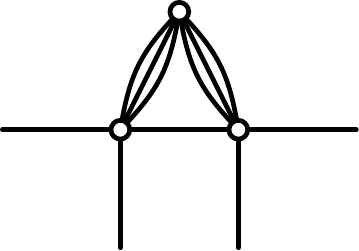}
    \caption{
      Simulating a degree-4 unbreakable vertex
      using degree-6 breakable vertices.
    }
    \label{fig:trvb6}
  \end{figure}
\end{proof}

\section{Hamiltonian Cycles in Grid Graphs}
\label{sec:grid graphs}

In this section, we prove ASP-completeness of finding Hamiltonian cycles
in several natural classes of grid graphs.
We begin by defining the types of graph that appear in our results.

\begin{definition}
  A \defn{grid graph}
  is an induced subgraph of the square lattice.
  That is, its vertices are a subset of $\mathbb Z^2$,
  and it has an edge between each pair of vertices at distance $1$.
  In a \defn{directed grid graph}, each edge has a direction,
  so there is exactly one edge between each pair of vertices at distance $1$.
\end{definition}

\begin{definition}
  A \defn{rectangular} grid graph
  is one whose vertex set consists of all lattice points within a rectangle.
\end{definition}

\begin{definition}
  A graph is \defn{max-degree-3}
  if each of its vertices have degree at most 3.
\end{definition}

\begin{definition}
  A \defn{spanning subgraph} of $G$
  is a subgraph of $G$ which contains all of the vertices
  (and some subset of the edges) of $G$.
\end{definition}

Note that grid graphs contain all possible edges:
graphs that contain only some of the edges
are (spanning) subgraphs of grid graphs.

We consider three types of graph
for each of undirected and undirected.
Our results are summarized in Table~\ref{tbl:summary}.

Most of our ASP-completeness results are by reductions
from planar $(\{4\},\{1\})$-TRVB, and use the same core idea illustrated in Figure~\ref{fig:exvertex}.
This is a breakable degree-8 vertex, with the yellow
square in the middle representing the vertex itself
and the blue tentacles representing edges.
We replace every vertex in the TRVB instance with a vertex like the one shown, and connect the tentacles of adjacent vertices.
By Lemma~\ref{lem:connections}, Hamiltonian cycles of the resulting graph correspond to solutions of the original TRVB instance.

\begin{figure}
  \centering
  \includegratility[trim=15 15 15 15]{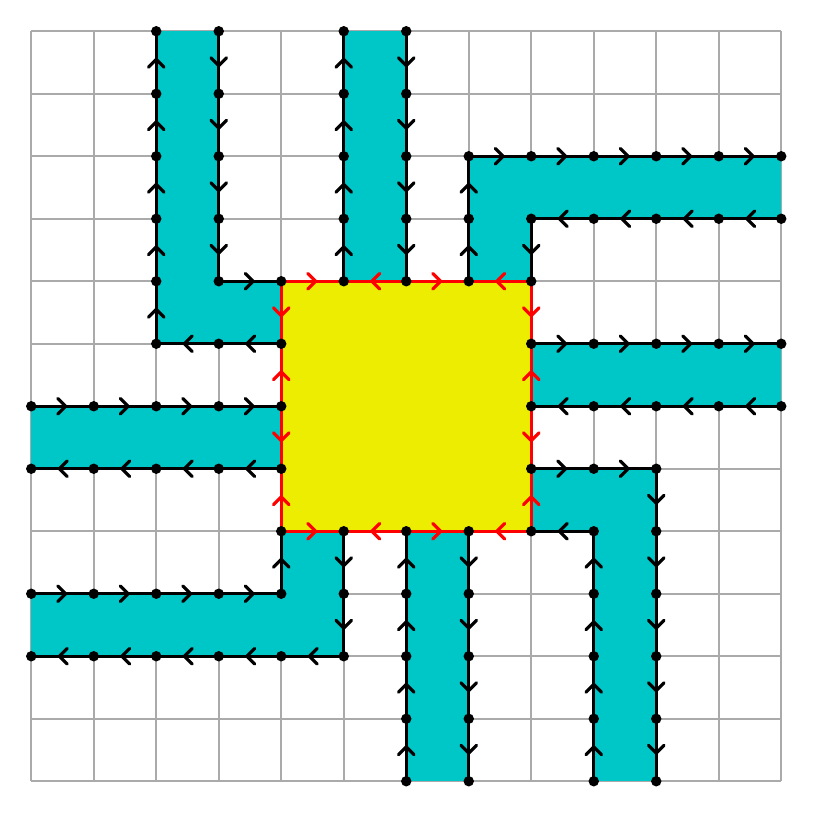}
  \caption{
    An example showing how reductions from TRVB
    to Hamiltonian cycle work.
  }
  \label{fig:exvertex}
\end{figure}

This idea works equally well for directed and undirected graphs.
To apply this idea to each of the five types of graph we prove ASP-completeness for,
we need to show how to draw gadgets for degree-4 breakable and degree-1 unbreakable vertices
in that type of graph,
while ensuring that the tentacles representing edges do not interfere with each other.

\subsection{Rectangular Grid Graphs}
\label{sec:rectangular}

\begin{theorem}[\cite{undir_rect}]
  Finding Hamiltonian cycles on an undirected rectangular grid graph is in P.
\end{theorem}

\begin{theorem}
  \label{thm:directed rectangular}
  Finding Hamiltonian cycles on a directed rectangular grid graph is ASP-complete.
\end{theorem}

\begin{proof}

  We first consider directed grid graphs, and later fill in holes to make them rectangular.
  Everything we need for this is shown in Figure~\ref{fig:dirrect}.
  The yellow rectangles are degree-4 breakable vertices
  with exactly two local solutions,
  and the dead end in the bottom left is a degree-1 unbreakable vertex.
  As before, blue is inside the loop and yellow might be
  inside the loop depending on the choice made for a vertex gadget.
  If we ignore the gray edges, this is essentially the same as Figure~\ref{fig:exvertex}.

  \begin{figure}
    \centering
    \includegratility[trim=15 30 15 15]{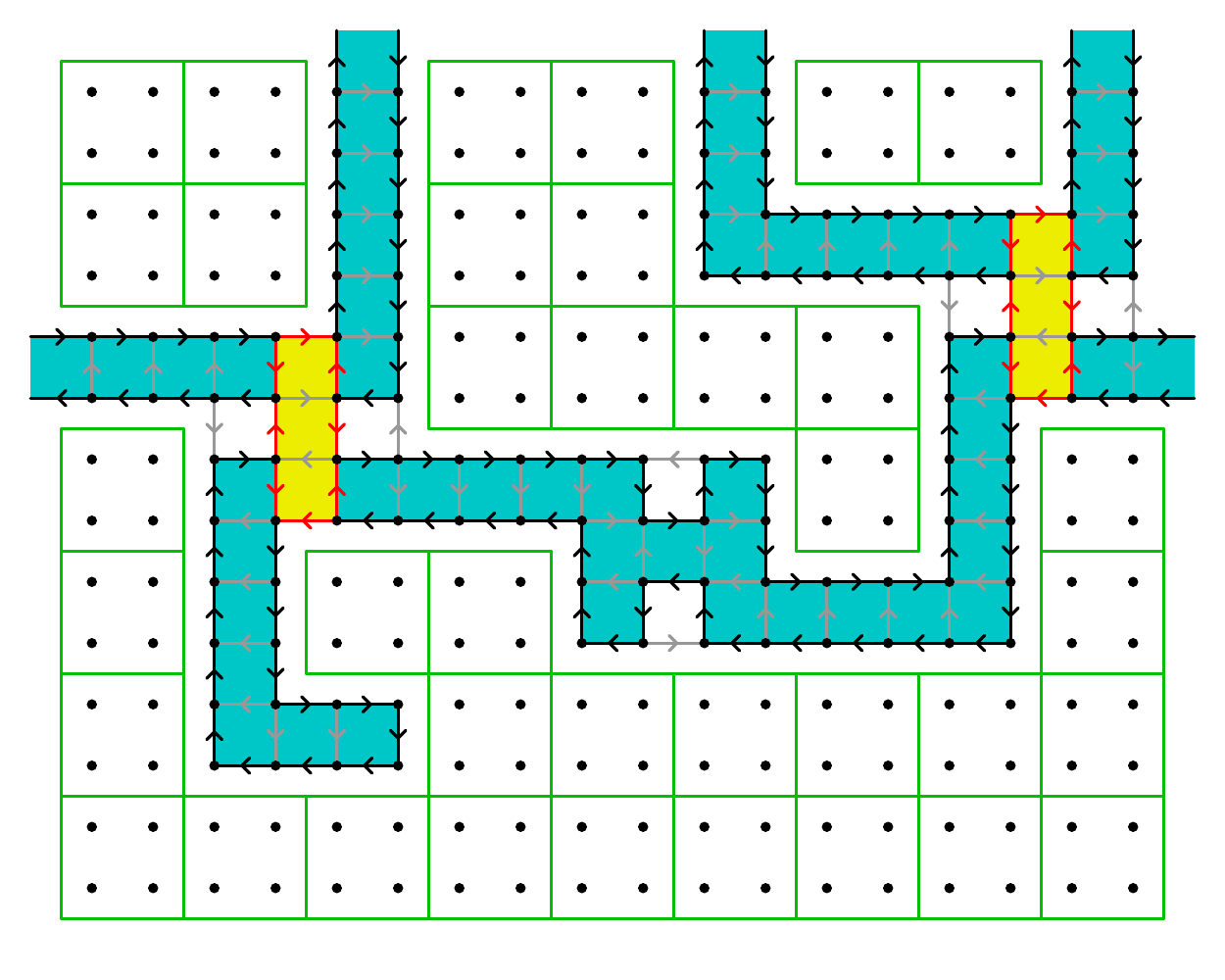}
    \caption{
      TRVB gadgets for directed grid graphs, showing
      two breakable degree-4 vertices connected by an edge
      and an unbreakable degree-1 vertex.
    }
    \label{fig:dirrect}
  \end{figure}

  We just need to ensure that gray edges cannot be used,
  which we can do by orienting them carefully.
  Ignoring the H-shaped construction in the center for the moment,
  each black edge is either the only edge pointing towards or the only edge
  pointing away from some vertex (depending on which side of the
  tentacle it's on), and thus must be used in a Hamiltonian cycle.
  We call such an edge \defn{forced}.
  Each gray edge (still ignoring the H) shares either its source or its target
  with a black edge, and thus cannot be used.
  We call such an edge \defn{unusable}.

  This requires the orientation of the gray edges relative to a tentacle
  to be different on the two ends of the tentacle,
  which is what the H achieves:
  one can verify by repeatedly finding forced edges
  and deleting unusuable edges that any Hamiltonian cycle
  must use all black edges and no gray edges in the H.
  Each tentacle representing an edge between two degree-4 breakable vertices
  will have such an H.

  This reduction proves a weaker version of the theorem:
  Finding Hamiltonian cycles on a directed grid graph is ASP-complete.
  It remains to fill all of the unused space to make a rectangular grid graph.

  If we place each vertex gadget, H, and turn on the same parity,
  the construction lies neatly on a $2\times 2$ grid,
  and in particular the holes are made of $2\times 2$ squares.
  Figure~\ref{fig:dirrect} indicates these squares in green.
  In addition, in each hole at least one of these squares
  is adjacent to a forced edge:
  all black edges except a few in each H are forced,%
  \footnote{
    They all become forced after deleting some unusable edges,
    but it's simpler to argue that hole filling works
    with directly forced edges.
  }
  and each hole is adjacent to a non-H section of tentacle
  provided we do not use any extremely short tentacles.

  Pick one such $2\times2$ square, and add four new vertices to fill it.
  Assume that the adjacent forced edge is the only outgoing edge from its source;
  the case where it is the only edge pointing towards its target is similar but with directions reversed.
  This situation is illustrated in Figure~\ref{fig:holefilling} (left),
  with the forced edge in blue.
  Now reverse the forced edge, and add new edges
  as shown on the right of Figure~\ref{fig:holefilling} (omitting any edges between a vertex
  in the square and a vertex outside it which doesn't yet exist).
  It is straightforward to check that all gray edges are unusable,
  so any Hamiltonian cycle must follow the blue path, which is equivalent
  to the original forced edge but consumes the added vertices.

  \begin{figure}
    \centering
    \includegratility[trim=15 15 15 15]{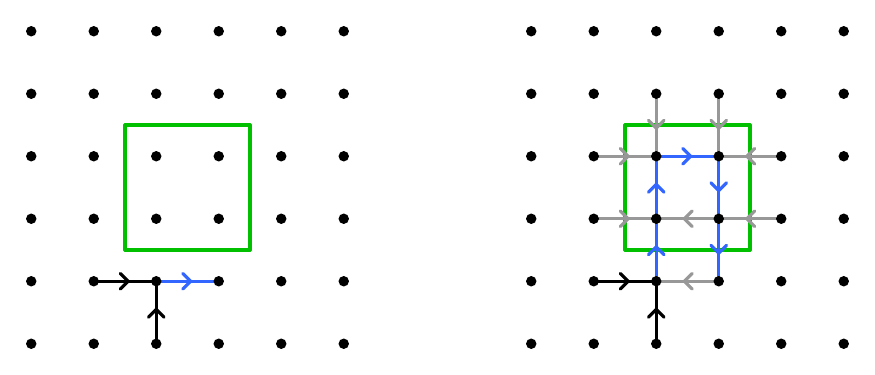}
    \caption{
      Filling holes in a directed rectangular grid graph.
    }
    \label{fig:holefilling}
  \end{figure}

  Filling this small portion of hole preserves the fact
  that every hole has a $2\times 2$ square adjacent to a forced edge,
  since the three relevant blue edges are forced.
  Thus we can repeat this process until all holes are filled,
  ultimately filling each hole with paths that outline a spanning forest
  of the $2\times 2$ squares.
  Figure~\ref{fig:filled} shows what this looks like
  after filling (the visible portion of) the top middle
  hole in Figure~\ref{fig:dirrect}.

  \begin{figure}
    \centering
    \includegratility[trim=15 30 15 15]{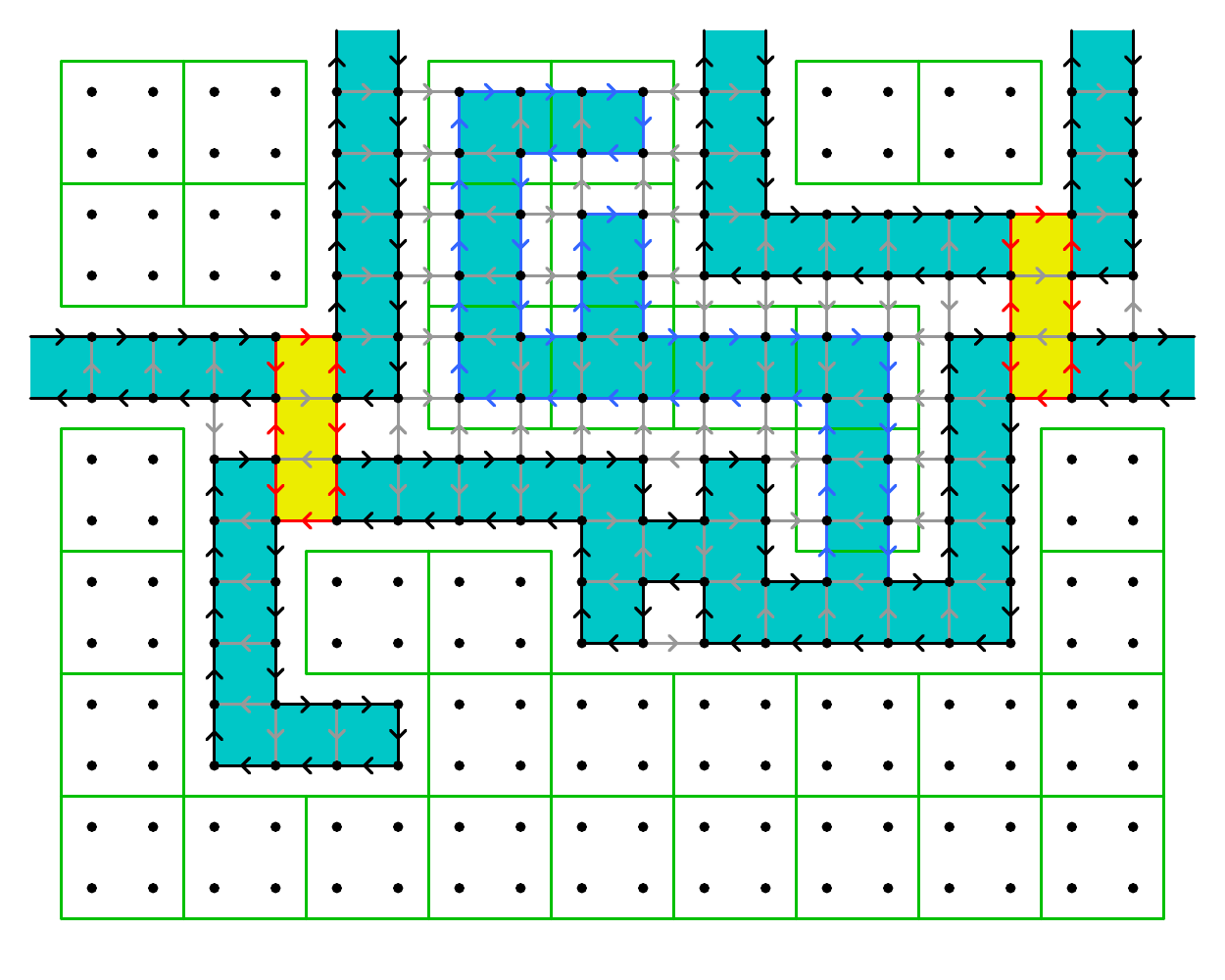}
    \caption{
      Figure~\ref{fig:dirrect} after some hole filling.
    }
    \label{fig:filled}
  \end{figure}

  The result is a directed rectangular grid graph which is equivalent
  to the original directed grid graph for the purposes of Hamiltonian cycles.
  Hence Hamiltonian cycles in the final graph correspond
  to solutions to the instance of TRVB.
\end{proof}

\subsection{Max-Degree-3 Spanning Subgraphs of Rectangular Grid Graphs}
\label{sec:spanning}

\begin{theorem}
  \label{thm:promise spanning}
  Let $G$ be a directed max-degree-3 spanning subgraph of a rectangular grid graph.
  Consider the promise problem of finding an \emph{undirected} Hamiltonian cycle on $G$,
  subject to the promise that all such cycles respect the given edge directions;
  that is, they would also be valid directed Hamiltonian cycles of $G$.
  This promise problem is ASP-complete.
\end{theorem}
\begin{proof}
  We modify the construction from Theorem~\ref{thm:directed rectangular} by simply removing all of the gray edges.
  Inspection of Figure~\ref{fig:filled} reveals that
  every vertex is incident to at most three non-gray edges:
  vertices along tentacles have two forced edges,
  and vertices in degree-4 vertex gadgets have one forced edge
  and two optional red edges.
  Filling holes preserves the non-gray degree of existing vertices
  and adds vertices with two non-gray edges.

  In the previous proofs, all of the possible solutions only used non-gray edges. Thus,
  we can adapt the previous reduction by simply deleting all gray edges,
  obtaining a directed max-degree-3 spanning subgraph of a rectangular grid graph.
  For instance, doing this to Figure~\ref{fig:filled}
  yields Figure~\ref{fig:cleaned}, which also has the advantage
  of being easier to read.

  By the proof of Lemma~\ref{lem:connections}, directed Hamiltonian cycles on $G$ are the same as undirected Hamiltonian cycles on $G$,
  and the set of such cycles is in bijection with solutions of the original TRVB instance.
\end{proof}

\begin{figure}
  \centering
  \includegratility[trim=15 30 15 15]{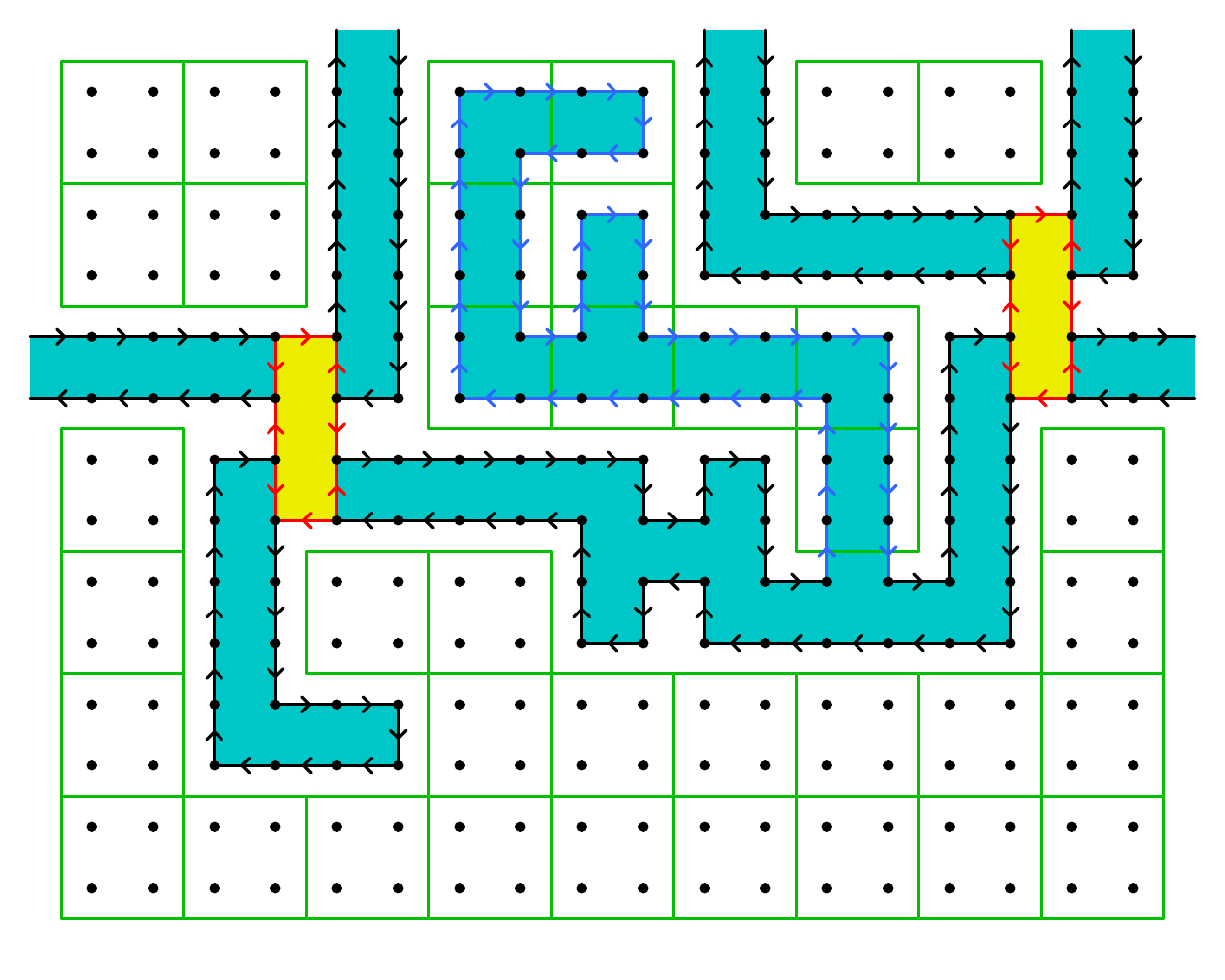}
  \caption{
    Figure~\ref{fig:filled} after removing gray edges.
  }
  \label{fig:cleaned}
\end{figure}

\begin{figure}
  \centering
  \includegratility[trim=15 30 15 15]{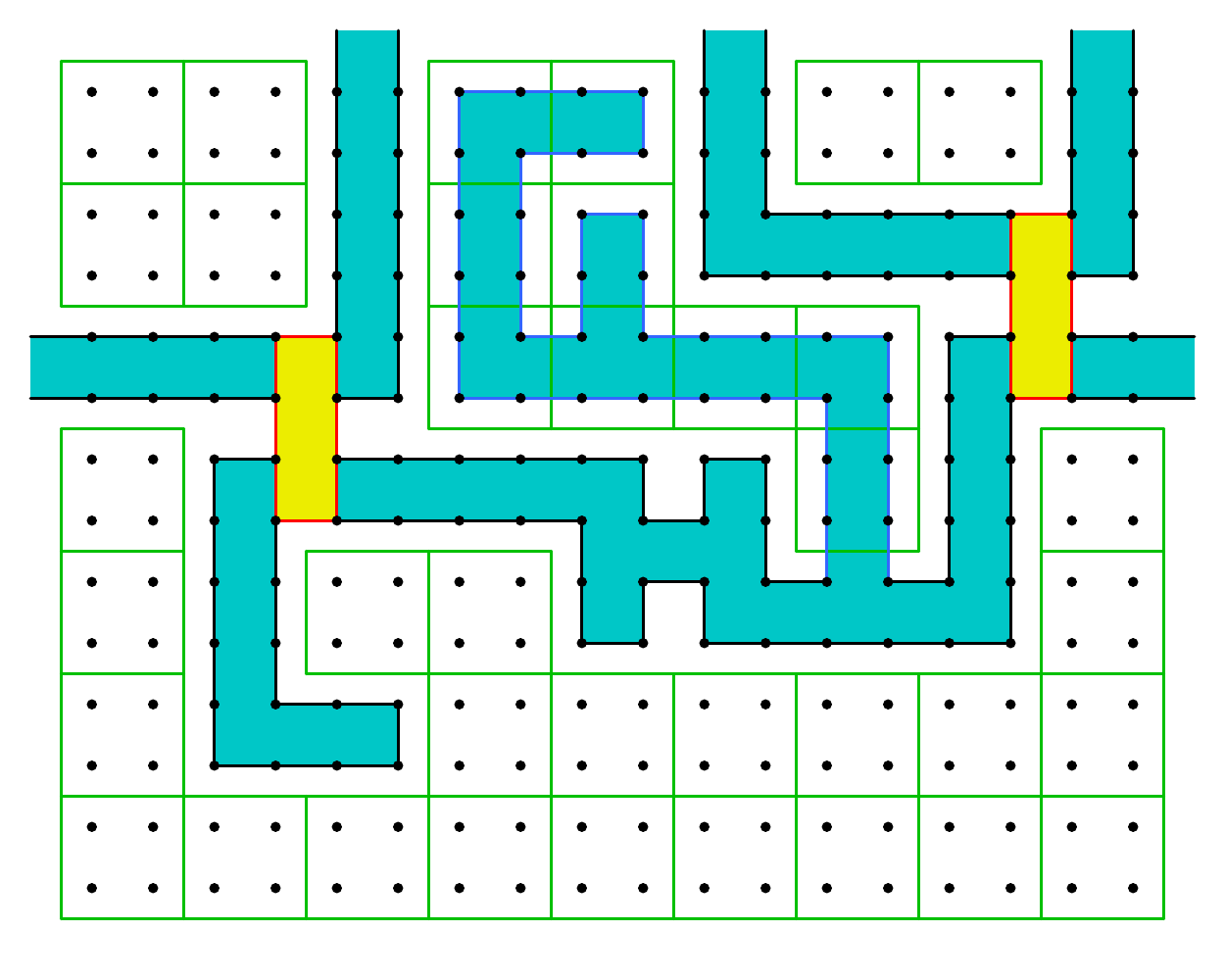}
  \caption{
    Figure~\ref{fig:cleaned} after forgetting directions of edges.
  }
  \label{fig:dedirected}
\end{figure}

\begin{corollary}
  \label{cor:directed spanning}
  Finding Hamiltonian cycles on a directed max-degree-3 spanning subgraph of a rectangular grid graph is ASP-complete.
\end{corollary}
\begin{proof}
  This is a special case of Theorem~\ref{thm:promise spanning}.
\end{proof}

In the undirected case, we can strengthen the assumption about forced edges.
For undirected graphs, an edge is \defn{forced} if it is incident to a degree-2 vertex,
since both edges incident to such a vertex must be used in any Hamiltonian cycle.
A degree-3 vertex in a subgraph of a grid graph has two edges in opposite directions,
which we call \defn{side} edges,
and a third edge between them, which we call the \defn{center} edge.
In this case, we can assume not only that each degree-3 vertex has a forced edge,
but that this forced edge is a side edge,
further reducing the number of distinct vertices we need to simulate for an application.

\begin{theorem}
  \label{thm:undirected spanning asymmetric}
  Finding Hamiltonian cycles on an undirected max-degree-3 spanning subgraph of a rectangular grid graph is ASP-complete,
  even when every degree-3 vertex has a forced side edge.
\end{theorem}
\begin{proof}
  We are not able to directly build breakable degree-4 TRVB vertices under these constraints.
  However, we are able to build a breakable degree-6 vertex,
  so we reduce from planar $(\{6\},\emptyset)$-TRVB,
  which was shown ASP-complete in Theorem~\ref{thm:trvb6}.

  Our breakable degree-6 vertex gadget is shown in Figure~\ref{fig:undir6}.
  Black edges are forced, and red edges are optional.
  Note that vertices in tentacles all have degree 2,
  and each degree-3 vertex inside the vertex gadget
  has a forced side edge. This is equivalent
  to the cycle of red edges turning at every vertex.
  The vertex gadget has exactly two local solutions,
  which each use alternating red edges.

  As before, blue tentacles are inside the cycle,
  and the yellow region is inside the cycle in one of the local solutions,
  corresponding to not breaking the TRVB vertex.
  We have new color as well: the green squares are inside the cycle
  in the other solution, when the TRVB vertex is broken.
  It is clear by inspection that the yellow local solution
  connects all six tentacles, and the green local solution disconnects them all.

  \begin{figure}
    \centering
    \includegratility[trim=15 15 15 15]{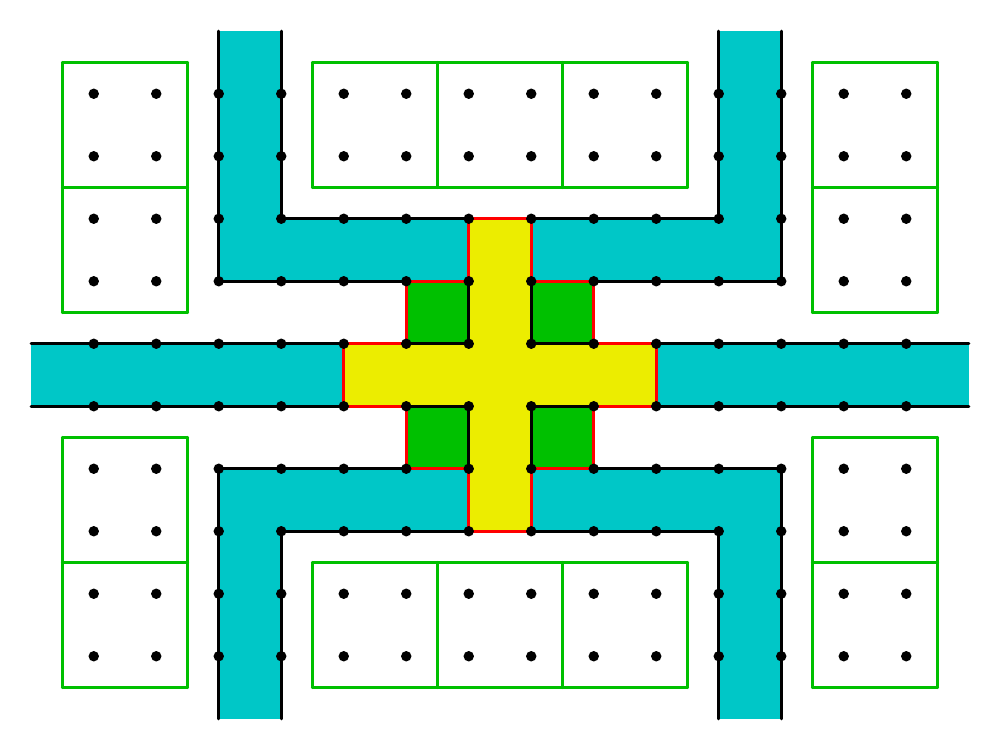}
    \caption{
      A breakable degree-6 TRVB vertex gadget for undirected
      max-degree-3 spanning subgraphs of rectangular grid graphs.
    }
    \label{fig:undir6}
  \end{figure}

  Finally, we connect vertex gadgets along tentacles and fill holes
  in exactly the same way as before.
  Filling holes uses only degree-2 vertices,
  so it does not introduce degree-3 vertices without forced side edges.
\end{proof}

\subsection{Max-Degree-3 Grid Graphs}
\label{sec:induced}

The existing proof of NP-hardness for finding Hamiltonian cycles
in max-degree-3 grid graphs \cite{papadimitriou} has only
one nonparsimonious gadget, the `fork connection'.
However, the reduction can be simplified
to avoid this gadget, making it parsimonious.
We will sketch the parsimonious version,
but we will also provide a different proof using TRVB
which yields the useful property that every vertex has a forced edge.

For the parsimonious adaptation of Papadimitriou and Vazirani's \cite{papadimitriou} proof,
we reduce from finding Hamiltonian cycles on non-alternating indegree-2 outdegree-2 planar graphs,
which is ASP-complete (Theorem~\ref{thm:in2out2}) \cite{trvb}.
See Figure~\ref{fig:dumbbell}.
Given such a graph, replace each vertex with a \defn{dumbbell},
which consists of two small loops called \emph{in} and \emph{out} connected by a path.
Each edge is represented by a \defn{tentacle} which connects
the out-loop of its source to the in-loop of its target.
These connections are slightly different,
and are shown with blue and red tentacles, respectively,
in Figure~\ref{fig:dumbbell}.

There is a lot of freedom in the placement of dumbbells,
but the parity of the position of each loop is important:
for tentacles to connect properly,
the in-loop must have white corners and the out-loop must have black corners.

\begin{figure}
  \centering
  \includegraphics[width=.3\linewidth,angle=180,trim=15 15 15 15]{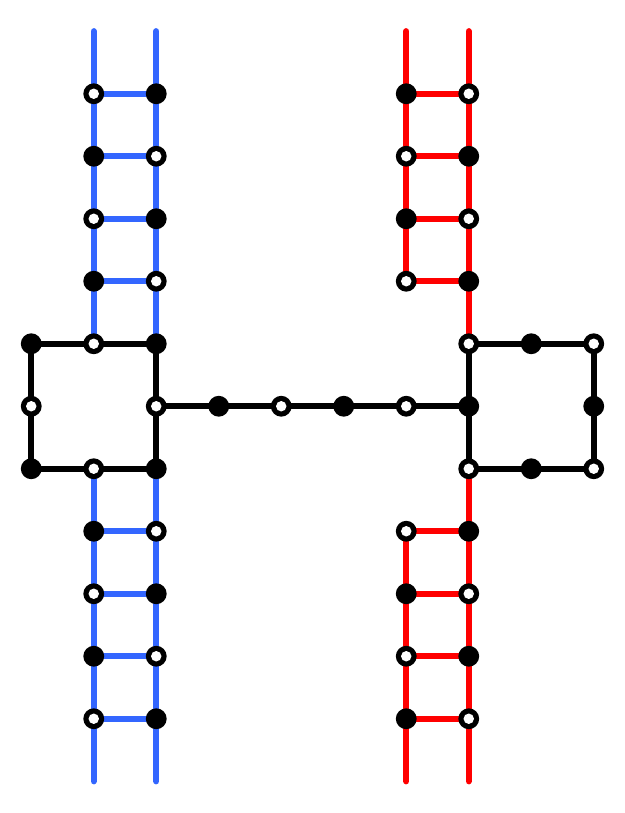}
  \caption{
    A dumbbell (black) with two red tentacles attached to its in-loop
    and two blue tentacles attached to its out-loop.
    Vertices are colored black and white in a checkerboard.
  }
  \label{fig:dumbbell}
\end{figure}

This works because each tentacle has only two solutions:
one in which the Hamiltonian cycle zigzags along it,
and one in which it loops back to the out-loop.
These correspond to the edge being used and unused, respectively.
Any Hamiltonian cycle in the grid graph must:
arrive at a dumbbell from a red tentacle,
go around the in-loop, cross the dumbbell to the out-loop,
go down and back one of the blue tentacles back to this dumbbell,
go around the out-loop to the other blue tentacle,
zigzag down that blue tentacle,
and finally arrive at the next dumbbell.
The sequence of dumbbells gives a Hamiltonian cycle on the original graph.

\begin{theorem} \label{thm:max deg 3 undir}
  Finding Hamiltonian cycles on an undirected max-degree-3 grid graph is ASP-complete,
  even when every vertex has a forced edge.
\end{theorem}
\begin{proof}
  This proof is sketched, and its key gadget is shown, by Demaine and Rudoy \cite{trvb},
  but at the time TRVB was not known to be ASP-complete,
  so it was purely a simpler proof of NP-hardness
  used to motivate the usefulness of TRVB.

  Like most of our other proofs, we reduce from planar $(\{4\},\{1\})$-TRVB.
  Our breakable degree-4 vertex gadget is shown in Figure~\ref{fig:undirgrid}.
  The main difficulty in this case
  is that we need the paths on each side of a tentacle to be separated
  by distance at least 2,
  so that the cycle cannot cross between the two sides
  (and all tentacle edges are forced).
  As usual, black edges are forced, and there are exactly two solutions
  which each use alternating red edges.
  One solution puts the green region inside the cycle,
  and one puts the yellow region inside the cycle,
  corresponding to breaking and not breaking the vertex, respectively.

  \begin{figure}
    \centering
    \includegratility[trim=15 15 15 15]{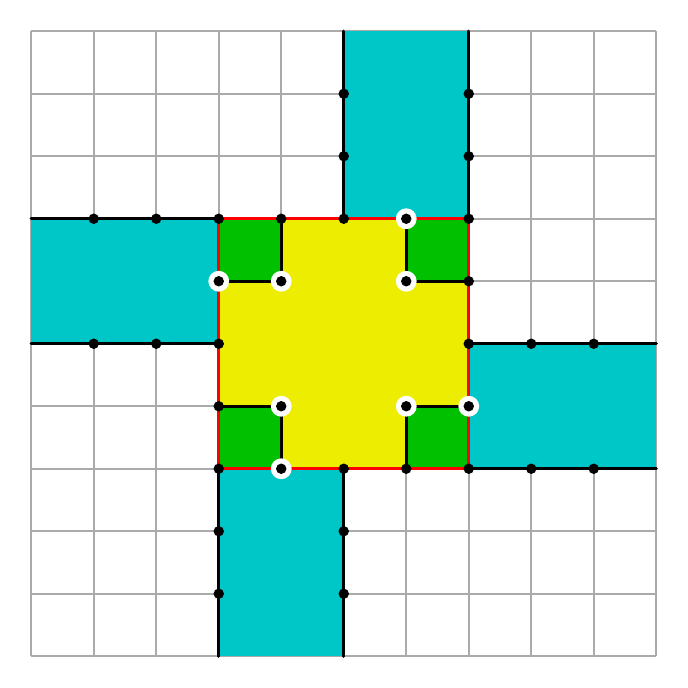}
    \caption{
      A breakable degree-4 TRVB vertex gadget
      for undirected max-degree-3 grid graphs.
      Removing the vertices highlighted in white gives an unbreakable degree-4 vertex gadget.
    }
    \label{fig:undirgrid}
  \end{figure}

  A degree-1 unbreakable vertex can be made by simply `capping off' a tentacle.
  Alternatively, we could reduce from $(\{4\},\{4\})$-TRVB,
  and construct a degree-4 unbreakable vertex gadget
  by removing the vertices highlighted in white from Figure~\ref{fig:undirgrid}.
\end{proof}

\begin{theorem} \label{thm:max deg 3 dir}
  Finding Hamiltonian cycles on a directed max-degree-3 grid graph is ASP-complete,
  even when every vertex has a forced edge.
\end{theorem}
\begin{proof}
  The proof is extremely similar to the previous proof.
  We again reduce from $(\{4\},\{1\})$-TRVB.
  Our degree-4 breakable vertex gadget is shown in Figure~\ref{fig:dirgrid},
  and a degree-1 unbreakable vertex can again be made by capping off a tentacle.
  Black edges are forced and gray edges are unusable.
  We again keep the sides of a tentacle apart from each other
  (away from vertex gadgets)
  so that a cycle cannot leak between them.

  \begin{figure}
    \centering
    \includegratility[trim=15 15 15 15]{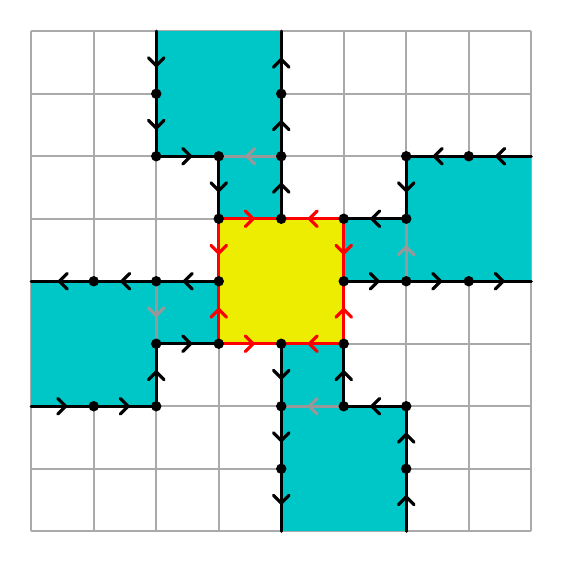}
    \caption{
      A breakable degree-4 TRVB vertex gadget
      for directed max-degree-3 grid graphs.
    }
    \label{fig:dirgrid}
  \end{figure}

  As before, there are exactly two solutions to the vertex gadget,
  one of which put the yellow square inside the cycle
  corresponding to leaving the TRVB vertex unbroken.
\end{proof}

\section{T-Metacells}
\label{sec:T-metacells}

Many puzzle genres which involve drawing a single loop are proven hard using reductions from various forms of grid graph Hamiltonicity. Tang \cite{tang} described a simple ``T-metacell'' framework for proving NP-hardness of these puzzles using grid graph Hamiltonicity.
A T-metacell is a gadget which represents a single degree-3 vertex in a grid graph. Each T-metacell is a (usually square) tile with 3 exits (on 3 of the 4 sides) such that the loop may traverse the gadget between any pair of exits. The gadget should be reflectable and rotatable, and the loop may travel between adjacent T-metacells only when both have exits along their shared border. Finally, the loop must be required to visit every T-metacell.

It's straightforward to see how T-metacells can simulate degree-3 vertices in a Hamiltonicity reduction; Tang showed that they can also simulate degree-2 vertices.
Let $G$ be a subgraph of a grid graph in which every vertex has degree 2 or 3.
Degree-3 vertices of $G$ can be replaced directly with T-metacells.
To handle degree-2 vertices, consider the graph $H$ on the same vertex set as $G$ which has an edge between two lattice-adjacent vertices precisely when $G$ is missing that edge.
Then $H$ consists of degree-1 and degree-2 vertices.
Orient the edges of $H$ into directed paths and cycles such that each vertex has a maximum indegree and outdegree of 1.
Each degree-2 vertex of $G$ can now be replaced by a T-metacell with its extra edge facing in the direction of the outward-pointing edge from that vertex in $H$.
This ensures that this extra exit will always be facing a non-exit in the adjacent cell, so only the intended edges of $G$ may be used by the loop.

We apply our results from Section~\ref{sec:grid graphs} to show that solving T-metacell problems is ASP-complete, instead of just NP-hard. We extend the framework to allow for some exits of a T-metacell to be \defn{directed}, meaning that the loop must have a consistent orientation which agree with the directions of the exits it uses.  We also allow for T-metacells to have one \defn{forced exit} through which the loop must pass. Note that when all three exits are directed, these necessarily create a forced exit: there must be either a lone exit directed inwards or a lone exit directed outwards, which in either case must be chosen.
T-metacells with forced edges can be classified into two categories: symmetric and asymmetric.
A symmetric T-metacell has its two unforced edges directly opposite each other, while an asymmetric T-metacell has its two unforced edges adjacent.
We use this classification to reduce the number of distinct gadgets which need to be constructed to apply the framework.

\begin{corollary} \label{cor:tundir}
  Finding Hamiltonian cycles on a rectangular grid of undirected T-metacells is ASP-complete.
\end{corollary}
\begin{proof}
  We reduce from finding Hamiltonian cycles on max-degree-3 spanning subgraphs of rectangular grid graphs (Theorem~\ref{thm:undirected spanning asymmetric}). Replace each vertex with a undirected T-metacell,
  handling degree-2 vertices as described above.
\end{proof}

\begin{corollary} \label{cor:tdirforced}
  Finding Hamiltonian cycles on a rectangular grid of required-edge directed T-metacells is ASP-complete.
\end{corollary}
\begin{proof}
  We reduce from finding Hamiltonian cycles on directed max-degree-3 spanning subgraphs of rectangular grid graphs (Corollary~\ref{cor:directed spanning}).
  Place a T-metacell for each degree-3 vertex,
  and handle degree-2 vertices in the same way as above.
  The direction of the unusable edge on a T-metacell at a degree-2 vertex
  can be arbitrary.
\end{proof}

\begin{corollary} \label{cor:tundirforced}
  Finding Hamiltonian cycles on a rectangular grid of asymmetric required-edge undirected T-metacells is ASP-complete.
\end{corollary}
\begin{proof}
  In the proof of Theorem~\ref{thm:undirected spanning asymmetric},
  every degree-3 vertex conveniently has a forced side edge,
  which is equivalent to being a asymmetric undirected T-metacell.
  Degree-2 vertices require a bit more care, but are not an obstruction:
  after deciding how to orient T-metacells as described above,
  note that for each degree-2 vertex,
  at least one of its edges is a side edge of the T-metacell.
  So we can simply place a T-metacell with that side edge forced.
\end{proof}

\begin{corollary} \label{cor:tmixedforced}
  Finding Hamiltonian cycles on a rectangular grid of
  required-edge directed asymmetric T-metacells
  and required-edge undirected symmetric T-metacells is ASP-complete.
\end{corollary}
\begin{proof}
  We reduce from the promise problem of finding a Hamiltonian cycle of a directed max-degree-3 spanning subgraph of a rectangular grid graph,
  with the promise that every undirected Hamiltonian cycle is a valid directed Hamiltonian cycle (Theorem~\ref{thm:promise spanning}).
  We perform the same replacement of vertices with T-metacells as in Corollary~\ref{cor:tdirforced}, except that the symmetric T-metacells are undirected.
  We claim that Hamiltonian cycles of the original graph are in bijection with solutions to the T-metacell instance.
  A directed Hamiltonian cycle of the original graph clearly solves the T-metacell instance, since it correctly passes through the directions on the directed T-metacells.
  On the other hand, a solution to the T-metacell instance is necessarily an undirected Hamiltonian cycle of the original graph;
  and by the promise, directed Hamiltonian cycles and undirected Hamiltonian cycles are the same.
\end{proof}

\section{Applications}
\label{sec:Applications}

In this section we apply our improved T-metacell framework to a variety of pencil-and-paper logic puzzles implemented by the online puzzle-solving interface ``puzz.link''
\cite{puzzlink}.
This web resource implements more than 240 different logic puzzles.
It includes most genres published by the Japanese publisher Nikoli,
whose puzzles have a long history of analysis from a computational complexity perspective
\cite{Seta02thecomplexities,Yato-Seta-2003,Andersson-2009,Uejima-pipelink,maarse,tang},
as well as many others in a similar style.

We improve existing NP-hardness results for pencil-and-paper logic puzzles to ASP-completeness, and give new ASP-completeness results.
Many of the ASP-completeness proofs consist of just a single T-metacell,
demonstrating the ease of applying the framework for proving ASP-completeness.
The main additional requirement when designing a T-metacell gadget for ASP-completeness proofs is that it be ``parsimonious'': for each pair of exits, there must be a \emph{unique} local solution where the loop passes through those exits.

\subsection{Loop-Drawing Paper-and-Pencil Logic Puzzles}
\newcommand\puzdef[1]{\defn{#1}}
\newcommand\puzpng[1]{\begin{center}\includegraphics[scale=0.25]{tcells/#1.png}\end{center}}
\newcommand\puzpdfs[2][]{\includegraphics[scale=0.5,#1]{puzzlink/#2.pdf}}
\newcommand\puzpdf[2][]{\begin{center}\puzpdfs[#1]{#2}\end{center}}
\newcommand\tc[1]{%
We construct a T-metacell to show ASP-completeness of #1
by Corollary~\ref{cor:tundir}.%
}
\newcommand\tcf[1]{%
We construct a set of undirected forced-edge T-metacells
to show ASP-completeness of #1
by Corollary~\ref{cor:tundirforced}.%
}
\newcommand\notc[1]{%
We can embed any max-degree-3 spanning subgraph
of a rectangular grid graph
in #1
to get ASP-completeness directly from Theorem~\ref{thm:undirected spanning asymmetric}.
}
\newcommand\ez{%
There is a straightforward reduction from Hamiltonicity
in undirected grid graphs,
which by Theorem~\ref{thm:max deg 3 undir} is ASP-complete:
}

The rules of these logic puzzle genres
share many commonalities.
In order to avoid repetition in describing them,
we first define some terminology.

The vast majority of the puzzles we analyze are loop-drawing puzzles.
This is a natural setting to apply the T-metacell framework,
as drawing a cycle is already part of the rules.

\begin{itemize}

    \item
        All of the puzzles we analyze take place on a rectangular grid graph.
        The solver can draw \puzdef{segments} between centers of orthogonally adjacent cells.

    \item
        In a \puzdef{loop} genre,
        the goal is to draw a set of segments
        that form a single cycle.
        In geometric terms,
        they must form a simple closed curve.

    \item
        A \puzdef{full} loop genre
        is one in which the loop is required to visit every cell.

    \item
        A \puzdef{crossing} loop genre
        is one in which the loop may cross itself,
        violating the standard ``simple'' constraint.

    \item
        A \puzdef{directed} loop genre
        is one in which the loop is additionally given a direction.
        By default,
        loop puzzles are assumed to be undirected.

    \item
        Considering each individual cell
        and how it connects to its neighbors,
        there are three possibilities up to rotation:
        it is either a \puzdef{turn}
        (if it connects to one vertically adjacent cell and one horizontally adjacent cell),
        goes \puzdef{straight}
        (if it connects to two cells on opposite sides),
        or is \puzdef{unused}.

    \item
        We can also decompose the loop into \puzdef{lines},
        which are contiguous runs of segments in the same direction.
        A loop always alternates between horizontal and vertical lines.
        The length of a line is the number of segments it contains.

    \item
        Some puzzles divide the grid into a set of \puzdef{regions}.
        In these puzzles,
        each region is a set of orthogonally connected cells,
        and the set of regions is a partition of the set of cells in the grid.

\end{itemize}

Although most of our results apply to loop genres,
the methods also work for some other classes of puzzles:

\begin{itemize}

    \item
        In a \puzdef{path} genre,
        the goal is to draw a set of segments that trace a single path.
        The start and end points of the path may be given,
        or they may be to be determined by the solver.
        All considerations that apply to loop puzzles
        (full, crossing, etc.)
        also apply to path puzzles.

        The framework applies fairly directly to path puzzles,
        since it is almost always easy to construct a metacell
        that is forced to contain both endpoints of the path.
        This simulates the Hamiltonian cycle problem,
        as a path that starts in this metacell,
        visits every other metacell,
        and ends in the original metacell
        is the same as a cycle at the metacell level.

    \item
        In a \puzdef{shading} genre,
        the goal is to mark some subset of the grid cells as shaded
        to satisfy some set of constraints.

    \item
        The constraint found in shading genres that allows us to apply our framework
        is \puzdef{connectivity},
        which says that some set of cells
        (typically all of the shaded cells)
        must form a single connected component.

        If we can enforce that the shaded cells leave each T-metacell
        by at most two exits via other constraints,
        the set of shaded cells simulates a path,
        and our framework thus applies as described above.

    \item
        Some genres do not fit into an overarching archetype
        like loop, path, or shading.
        We refer to such puzzles as \puzdef{variety} puzzles,
        which occasionally include a set of constraints conducive to our framework.

\end{itemize}

\subsection{Prior ASP-Completeness Results}

Some of the T-metacell gadgets in \cite{tang} are already parsimonious,
and thus automatically serve as ASP-completeness proofs for their respective genres
given our new results.
These genres are
Slalom,
Onsen-meguri,
Mejilink,
Detour,
Tapa-Like Loop,
Kouchoku,
and Icelom.

For some other puzzle genres,
ASP-completeness proofs exist outside of the T-metacell framework.
Yato and Seta
\cite{Yato-Seta-2003}
prove that Slitherlink is ASP-complete
in the paper originally defining the ASP class.
Uejima, Suzuki, and Okada
\cite{Uejima-pipelink}
prove that Pipelink is ASP-complete,
which also proves ASP-completeness of the generalized versions
Pipelink Returns and Loop Special.

\subsection{Prior NP-Hardness Improved to ASP-completeness}

Most of the gadgets in this section consist of minor adjustments
to existing T-metacells in \cite{tang} to ensure parsimony.

\subsubsection{Masyu}

Masyu \cite{Friedman-masyu,tang} is a loop genre.
Some cells contain pearls,
which can be either black or white.
In any cell with a black pearl,
the loop must turn,
and it must go straight through both of the two adjacent cells.
In any cell with a white pearl,
the loop must go straight,
and it must turn in at least one of the two adjacent cells.

\tc{Masyu}
\puzpdf{masyu}

\subsubsection{Yajilin}

Yajilin \cite{Ishibashi,tang} is a loop genre.
Some cells contain numbered arrow clues,
which count the number of unused cells from the clue to the edge of the grid
(not including the clue itself).
The loop cannot go through the clues,
and of the remaining cells,
unused cells may not be orthogonally adjacent.

\ez
start with the bounding box of the grid graph as a rectangular grid graph,
place a ``0'' clue pointing in an arbitrary direction in every cell excluded from the original graph,
and then add a row above the grid entirely filled with ``0'' clues pointing down,
forcing every other cell to be visited.

\subsubsection{Nagareru}

Nagareru (a.k.a.\ Nagareru-Loop)
\cite{Iwamoto-moon-nagare-nuri,tang} is a directed loop genre.
Shaded cells cannot be visited.
Some unshaded cells contain arrows,
which must contain a straight segment and indicate the direction of the loop along that segment.
Some shaded cells contain arrows,
which exert a ``wind'' on all unshaded cells in the direction of the arrow
up to the next shaded cell.
Whenever the loop enters a cell experiencing wind,
it must immediately turn in the direction of the wind.

We construct a set of undirected forced-edge T-metacells to show ASP-completeness of Nagareru by Corollary~\ref{cor:tundirforced}.
Note that the resulting construction has two solutions
for every solution of the original undirected Hamiltonicity problem --
we must also fix a global direction
by inserting an arrow in the center cell of one T-metacell.
\puzpng{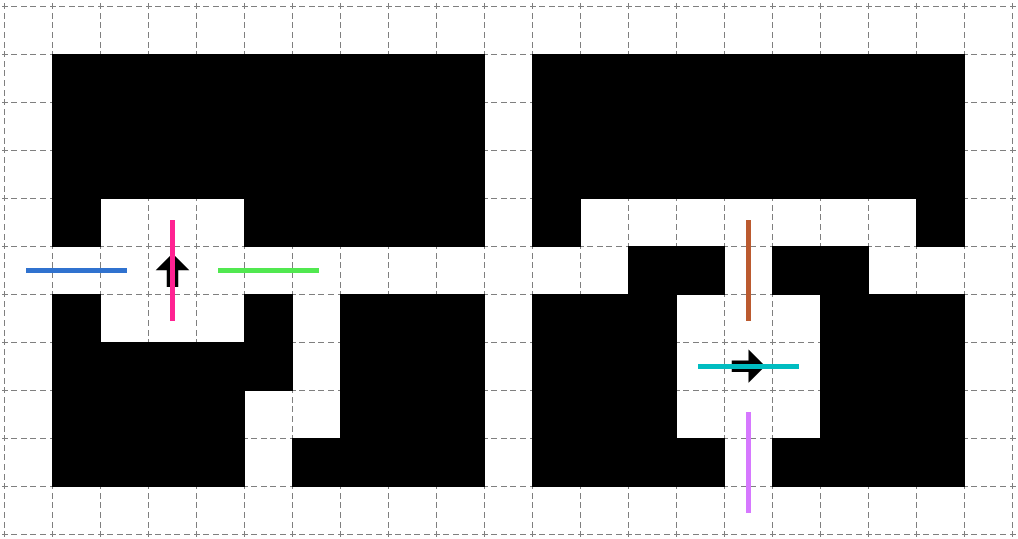}

We note that we can alternatively construct a set of directed forced-edge T-metacells to show ASP-completeness by Corollary~\ref{cor:tdirforced}.
\puzpng{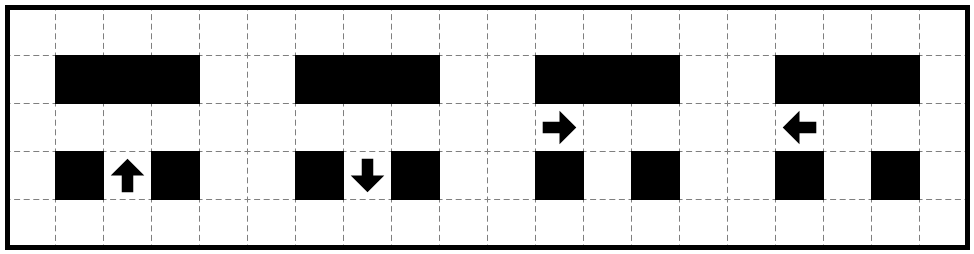}

\subsubsection{Castle Wall}

Castle Wall \cite{tang} is a loop genre.
Some cells are shaded white, black, or gray.
Shaded cells cannot be visited.
White cells must be contained in the interior of the loop,
and black cells must be in the exterior of the loop.
Additionally,
some shaded cells have numbered arrows,
which count the number of segments in the direction of the arrow
up to the edge of the grid
with the same orientation as the arrow (vertical or horizontal).

\tcf{Castle Wall}
\puzpng{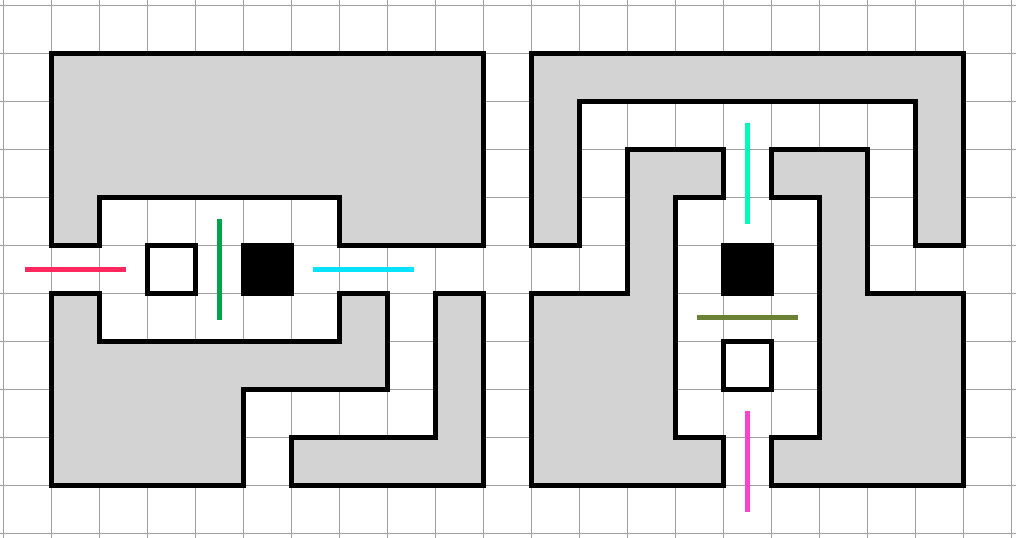}

We note that we can alternatively construct a set of directed forced-edge T-metacells to show ASP-completeness by Corollary~\ref{cor:tdirforced}.
This works because if the global orientation of the loop is arbitrarily thought of as travelling clockwise,
it always ``sees'' white squares on its right and black squares on its left,
and these T-metacells thereby force its direction.
\puzpng{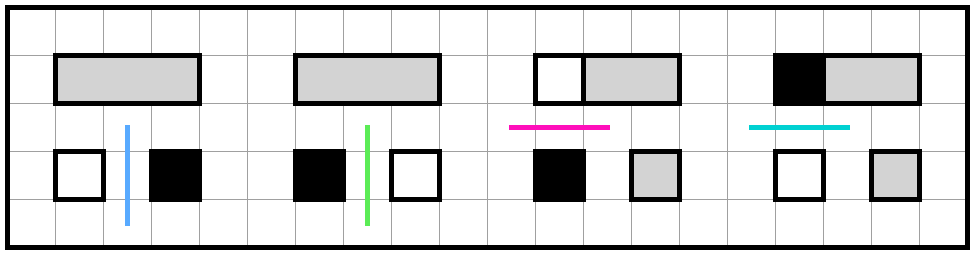}

\subsubsection{Moon or Sun}

Moon or Sun \cite{Iwamoto-moon-nagare-nuri,tang} is a loop genre with regions.
Some cells contain moons, and some cells contain suns.
The loop must visit each region exactly once,
and within each region must visit either all moons and no suns,
or all suns and no moons.
Furthermore,
the loop cannot visit the same symbol in two consecutively used regions,
so it must alternate between ``sun-regions'' and ``moon-regions.''

\ez
take the bounding box of the grid graph as a rectangular grid graph,
place a sun in every cell contained in the original graph,
place a moon in every other cell,
and finally replace any sun with a \(1\times1\) region containing a moon.

\subsubsection{Country Road}

Country Road \cite{Ishibashi,tang} is a loop genre with regions.
Some regions contain numbers,
which count the number of cells the loop passes through in that region.
Furthermore,
no pair of orthogonally adjacent cells in different regions can both be unused.

\tc{Country Road}
\puzpdf[scale=1.5]{country}

\subsubsection{Geradeweg}

Geradeweg \cite{tang} is a loop genre.
Some cells contain numbers,
which count the length of all lines touching that cell.
All numbers must be visited.

\tc{Geradeweg}
\puzpdf[scale=1.5]{geradeweg}

\subsubsection{Maxi Loop}

Maxi Loop \cite{tang} is a full loop genre with regions.
Some regions contain numbers,
which give the number of cells occupied by the maximal set of contiguous segments
contained in that region.

\tc{Maxi Loop}
\puzpdf{maxi}

\subsubsection{Mid-loop}

Mid-loop \cite{tang} is a loop genre.
Dots can be placed on cell borders or at the centers of cells,
and every dot must mark the midpoint of some line in the loop.

\tc{Mid-loop}
\puzpdf{midloop}

\subsubsection{Balance Loop}

Balance Loop \cite{tang} is a loop genre.
Some cells contain either black or white circles.
All circles must be visited.
Circles give information about the two lines emanating from the circle,
in the direction of each incident segment
(which may be both-horizontal or both-vertical).
For a white circle,
their lengths must be the same,
and for a black circle,
their lengths must be different.
Additionally,
if the circle has a number,
it gives the sum of the two lengths.

\tc{Balance Loop}
\puzpdf{balance}

\subsubsection{Simple Loop}

Simple Loop \cite{undir_rect,tang} is a loop genre.
Some cells are shaded and cannot be visited,
but all other cells must be visited.

Simple Loop is directly ASP-complete from Theorem~\ref{thm:max deg 3 undir}.

\subsubsection{Haisu}

Haisu \cite{tang-old,tang} is a path genre with regions.
Some cells contain numbers.
The first time the path visits a region,
it must visit all the 1s;
the second time the path visits a region,
it must visit all the 2s;
and so on.

\tc{Haisu}
Note that we must break the loop into a path
by e.g.\ picking an arbitrary T-metacell,
replacing its center cell with an S,
and replacing the cell to the right with a G.
\puzpdf{haisu}

\subsubsection{Reflect Link}

Reflect Link \cite{tang} is a crossing loop genre.
Some cells contain mirrors,
which must contain a turn in the depicted orientation.
Some mirrors contain numbers,
which are 1 greater than the sum of the two incident lines
(in other words,
they count the total number of cells occupied by the two incident lines).
The loop can only cross on marked crossings.

\tc{Reflect Link}
\puzpdf{reflect}

\subsubsection{Linesweeper}

Linesweeper \cite{maarse} is a loop genre.
Some cells contain numbers,
which count the total number of cells the loop visits
among the 8 orthogonally and diagonally adjacent cells.
Numbers must not be visited.

\tc{Linesweeper}
Note that puzz.link does not contain a Linesweeper implementation,
so it was not analyzed in Tang's paper;
the known NP-completeness result
was obtained independently
\cite{maarse}.
\puzpng{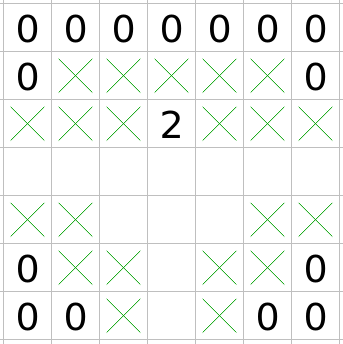}

\subsection{New NP- and ASP-Completeness Results}

\subsubsection{Vertex/Touch Slitherlink}

Vertex Slitherlink and Touch Slitherlink are loop genres.
The presentation differs slightly from most other loop genres
in that lines are drawn between dots
instead of between centers of cells.
Clues refer to the four surrounding vertices:
for Vertex Slitherlink,
they count the number of vertices visited by the loop,
and for Touch Slitherlink,
they count the number of distinct times the loop visits any of the four surrounding vertices
(where a segment between two of them does not count as a separate visit).

\tc{both versions}
\puzpdf{vslither}

\subsubsection{Dotchi-Loop}

Dotchi-Loop is a loop genre with regions.
Some cells contain black circles,
which mean they cannot be visited.
Some cells contain white circles,
which must be visited.
Additionally,
within each region,
the shape of the loop
(whether it is a turn or goes straight)
must be the same on all white circles.

\ez
take the bounding box of the grid graph as a rectangular grid graph,
place a \(1\times1\) region with a white circle in every cell contained in the original graph,
and place a black circle in every other cell.

\subsubsection{Ovotovata}

Ovotovata is a loop genre with regions.
Some regions contain numbers,
which for every time the loop exits that region in any direction
count the number of additional cells it travels.
Additionally,
some regions are shaded,
which means they must be visited.

\ez
take the bounding box of the grid graph as a rectangular grid graph,
place a \(1\times1\) shaded region in every cell contained in the original graph,
and place an unshaded region containing a number larger than the size of the grid in every other cell.

\subsubsection{Building Walk}

Building Walk is a path genre.
Some cells are shaded, representing elevators;
and some unshaded cells have numbers: every number and elevator must be visited.
Numbers on unshaded cells indicate which ``floor'' the path must be on when it reaches that number.
Whenever the path reaches an elevator, it must change floors, and it cannot change floors except at elevators.
Elevators may be marked with an arrow indicating whether the floor increases or decreases at that elevator.
The path cannot go below the 1st floor or above the $n$th floor, where $n$ is the maximum number on the grid.
The start and end of the path are designated by `S' and `G' on unshaded cells, which may also have a number indicating which floor the path starts or ends on.

We construct a set of $5 \times 5$ required-edge directed asymmetric T-metacells and a required-edge undirected symmetric T-metacell to show ASP-completeness of Building Walk (Corollary~\ref{cor:tmixedforced}).
\begin{center}
  \puzpdfs{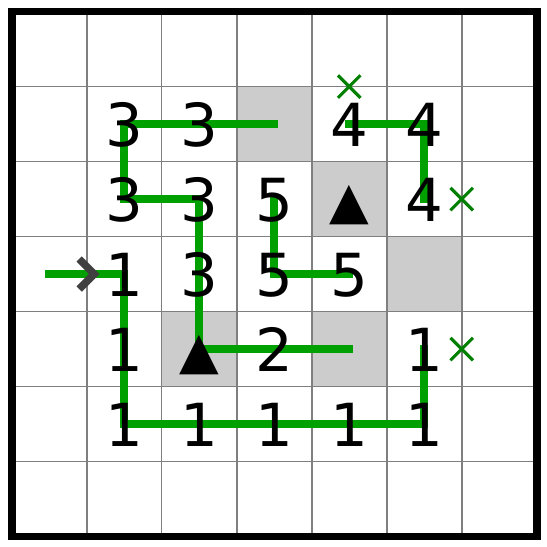}
  \puzpdfs{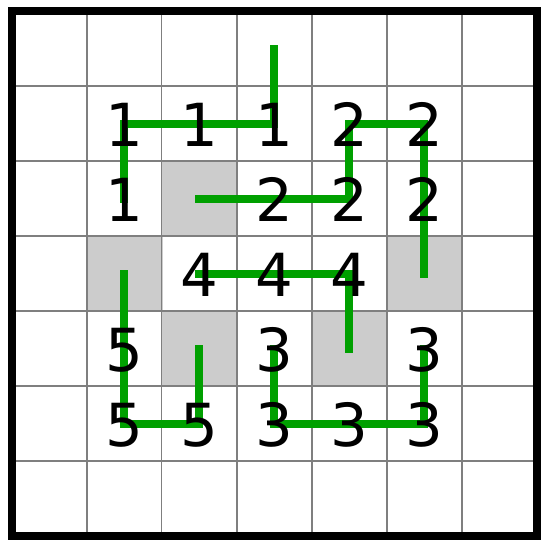}
\end{center}

For now we assume each number appears in at most one T-metacell, so that edges cannot be drawn between numbered cells in different T-metacells.  The left diagram shows a directed asymmetric T-metacell, which has a required edge on the left and cannot exit through the bottom.  Its orientation may be reversed by inverting the arrows on the elevators and applying the involution $x \mapsto 6 - x$ to its numbered cells.
The right diagram shows a required-edge undirected symmetric T-metacell.

There is one technical issue with these constructions, which is that two adjacent T-metacells whose required edges both point toward each other must share a number where the required edges connect.

If the two T-metacells' required-edge paths (i.e. the paths labeled ``1'' in the above diagrams) turn in opposite directions,
then we can just label the two paths with the same number.
However, if they turn towards the same direction,
then this might allow for unintended solutions.

The following diagram shows how to resolve the situation in this case.
It shows the boundary between two T-metacells whose required-edge paths (labeled ``1'' and ``6'') turn towards the same direction.
In this case, replacing two cells of the ``6'' path with a ``1'' and an elevator as shown forces the paths to connect properly without allowing extra solutions.

\begin{center}
  \puzpdfs{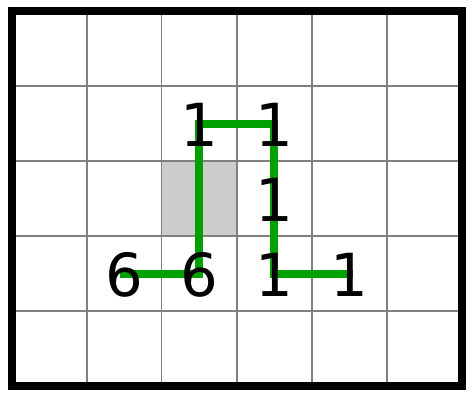}
\end{center}

Finally we must also break the loop into a path; this can easily be done by picking any two unshaded cells which are forced to be adjacent in the loop,
and labeling one with `S' and one with `G'.

\subsubsection{Rail Pool}

Rail Pool is a full loop genre with regions.
For each region,
consider all lines that overlap any cell of that region,
and take the set of their lengths.
This set must match the set of numbers in the region,
ignoring duplicates.
Unnumbered regions are unconstrained.

\tcf{Rail Pool}
\begin{center}\puzpdfs{railpool1} \puzpdfs{railpool2}\end{center}

\subsubsection{Disorderly Loop}

Disorderly Loop is a loop genre.
Each clue contains a multiset of numbers and points to a cell.
The clue cannot be visited,
the cell pointed to must contain a turn,
and the multiset contains the lengths of the next \(n\) lines
along the loop in the direction of the arrow.

\tc{Disorderly Loop}
\puzpdf{disloop}

\subsubsection{Ant Mill}

Ant Mill is a shading genre.
The goal is to draw a loop composed of dominoes of shaded cells,
connected by diagonal adjacency.
Some edges have squares,
which indicate that the two incident cells are either both shaded or both unshaded.
Some edges have X marks,
which indicate that exactly one of the two incident cells is shaded.

\tc{Ant Mill}
\puzpdf{antmill}

\subsubsection{Koburin}

Koburin is a loop genre.
Some cells contain numbers,
which count the number of orthogonally adjacent unused cells.
Numbers cannot be visited,
and unused cells cannot be orthogonally adjacent to each other.

\ez
take the bounding box of the grid graph as a rectangular grid graph,
and place a ``0'' clue in every cell excluded from the original graph.
Since the original graph has max degree 3,
every cell is adjacent to at least one unused cell,
and the clues therefore force every other cell to be visited.

\subsubsection{Mukkonn Enn}

Mukkonn Enn is a full loop genre.
Some cells contain numbers in one of their four quadrants,
which counts the length of the line extending from the cell in that direction
if a segment in that direction is present.
(If no such segment is present,
the number can be ignored.)

\tc{Mukkonn Enn}
\puzpdf{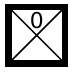}

\subsubsection{Rassi Silai}

Rassi Silai is a variety genre with regions.
Each region must contain exactly one path,
which visits all cells in that region.
Additionally,
no two endpoints of lines can be orthogonally or diagonally adjacent.%
\footnote{
    The standard rules of Rassi Silai also allow the grid to contain shaded cells,
    which cannot be visited,
    but this makes the reduction trivial from Theorem~\ref{thm:max deg 3 undir}.
}

\tcf{Rassi Silai}
Note that this T-metacell cannot be reflected,
as they rely on the borders present in adjacent cells,
so the first three images depict the three possible forced edges.
We must break the loop into a path
by replacing the top left cell with the fourth image shown.
\begin{center}
    \puzpdfs{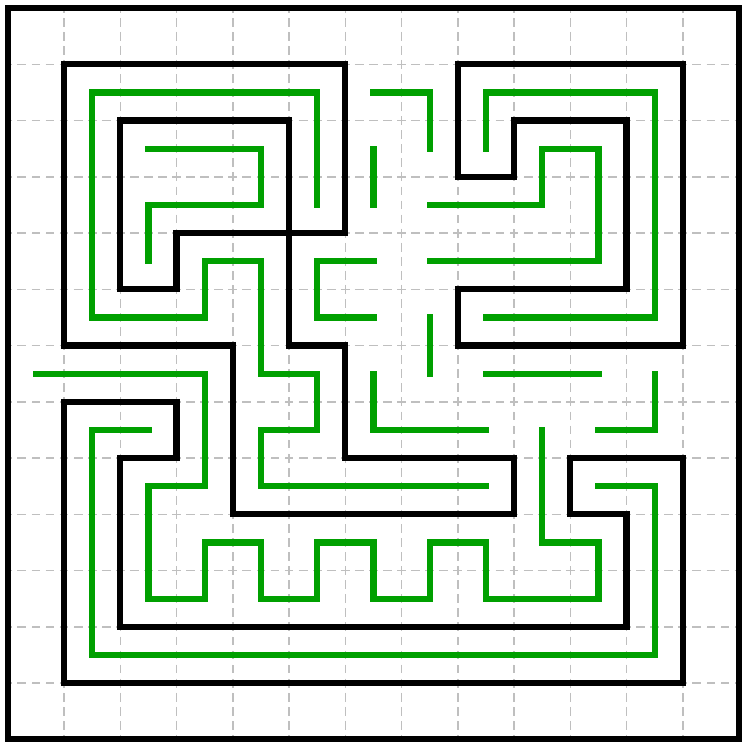}
    \puzpdfs{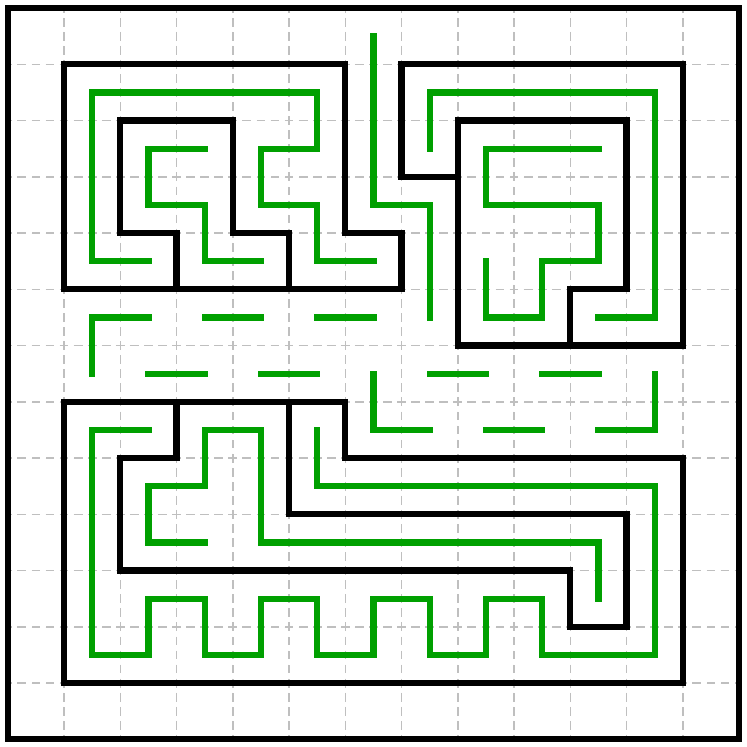}
    \puzpdfs{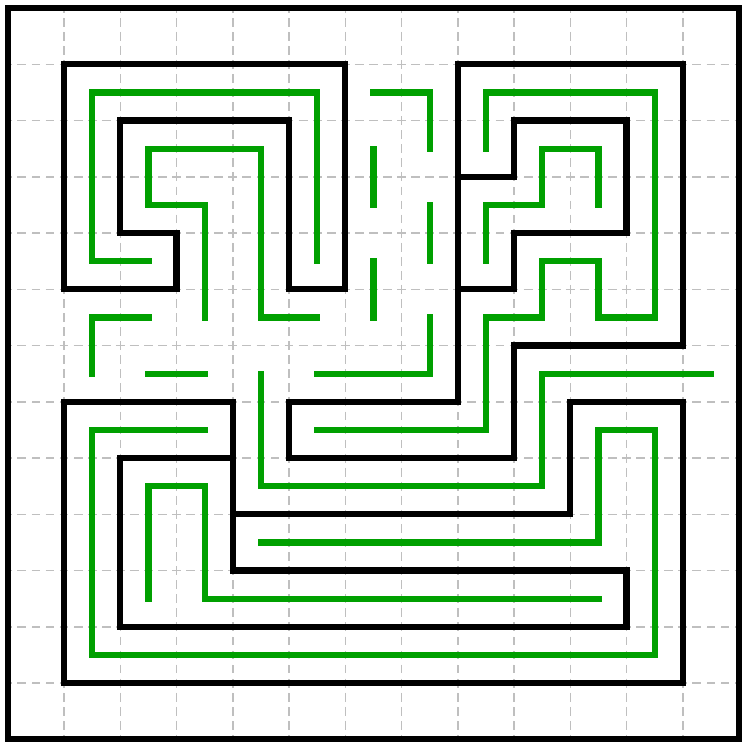}
    \puzpdfs{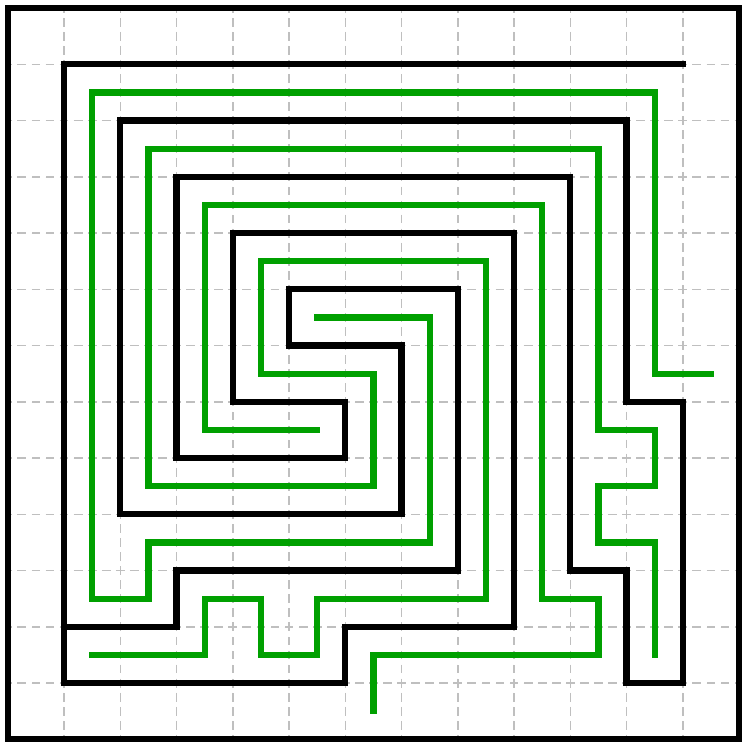}
\end{center}

\subsubsection{(Crossing) Ichimaga}

Ichimaga is a variety genre.
Some gridpoints contain circled numbers,
which count the number of adjacent segments used.
All drawn segments must form lines that connect two circles
and turn at most once.
Additionally,
the resulting graph of circles joined by lines
must be connected.
The Crossing Ichimaga variant allows these connecting lines to cross.

\tc{both versions}
Note that due to the global structure of the construction,
it is never possible for lines to turn or cross.
\begin{center}\puzpdfs{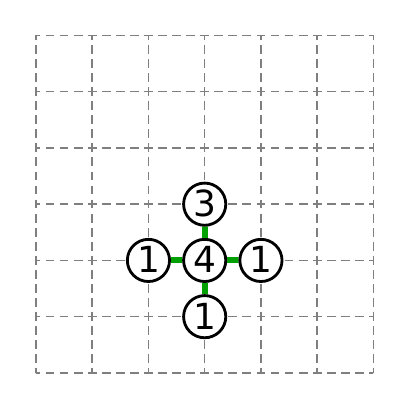} \puzpdfs{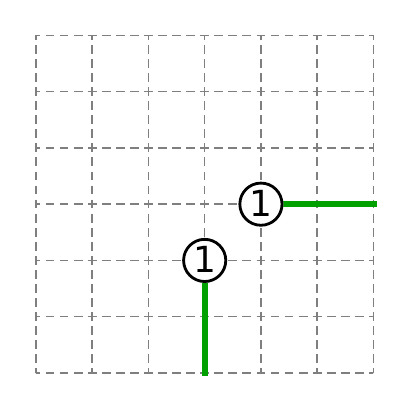}\end{center}

\subsubsection{Tapa}

Tapa is a shading genre.
Some cells contain multisets of numbers,
which must match the multiset of lengths of contiguous runs of shaded cells
in the 8 surrounding cells.
Additionally,
the shaded cells must be connected,
and \(2\times2\) squares of shaded cells are forbidden.

\tc{Tapa}
Note that we must break the loop into a path
by replacing the top left cell with the second image shown.
\begin{center}\puzpdfs{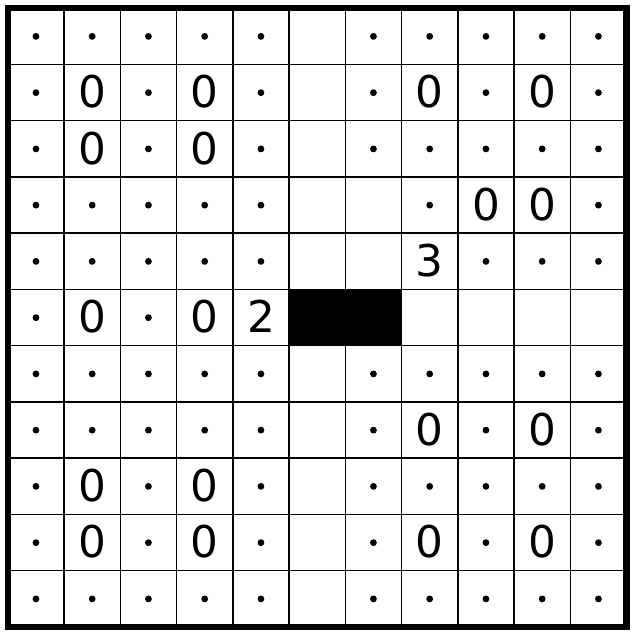} \puzpdfs{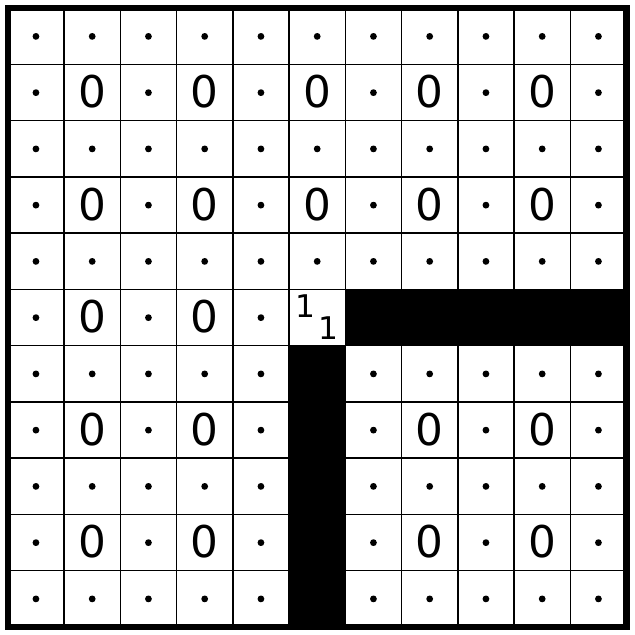}\end{center}

\subsubsection{Canal View}

Canal View is a shading genre.
Some cells contain numbers,
which the sum of lengths of runs of shaded cells
in all 4 orthogonal directions.
Additionally,
the shaded cells must be connected,
and \(2\times2\) squares of shaded cells are forbidden.

\notc{Canal View}
First,
expand every cell into the first image shown
to embed a rectangular grid graph.
Then,
remove the appropriate edges from each metacell
by placing ``0'' clues in any of the 4 cells orthogonally adjacent to the ``13'' clue.
Finally,
replace the top left cell with the second image shown
to break the loop into a path.
\begin{center}\puzpdfs{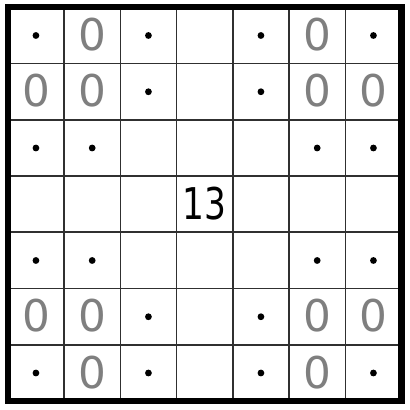} \puzpdfs{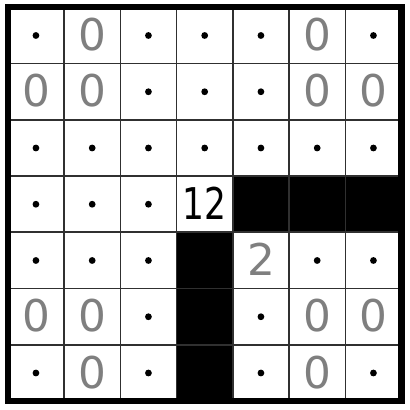}\end{center}

\subsubsection{Aqre}

Aqre is a shading genre with regions.
Some regions contain numbers,
which count the number of shaded cells in that region.
Additionally,
the shaded cells must be connected,
and no horizontal or vertical line of 4 cells in a row
can be all shaded or all unshaded.

\notc{Aqre}
First,
expand every cell into the first image shown
to embed a rectangular grid graph.
Then,
remove the appropriate edges from each metacell
by placing ``0'' clues in any of the 4 undetermined cells.
Finally,
replace the top left cell with the second image shown
to break the loop into a path.
\begin{center}\puzpdfs{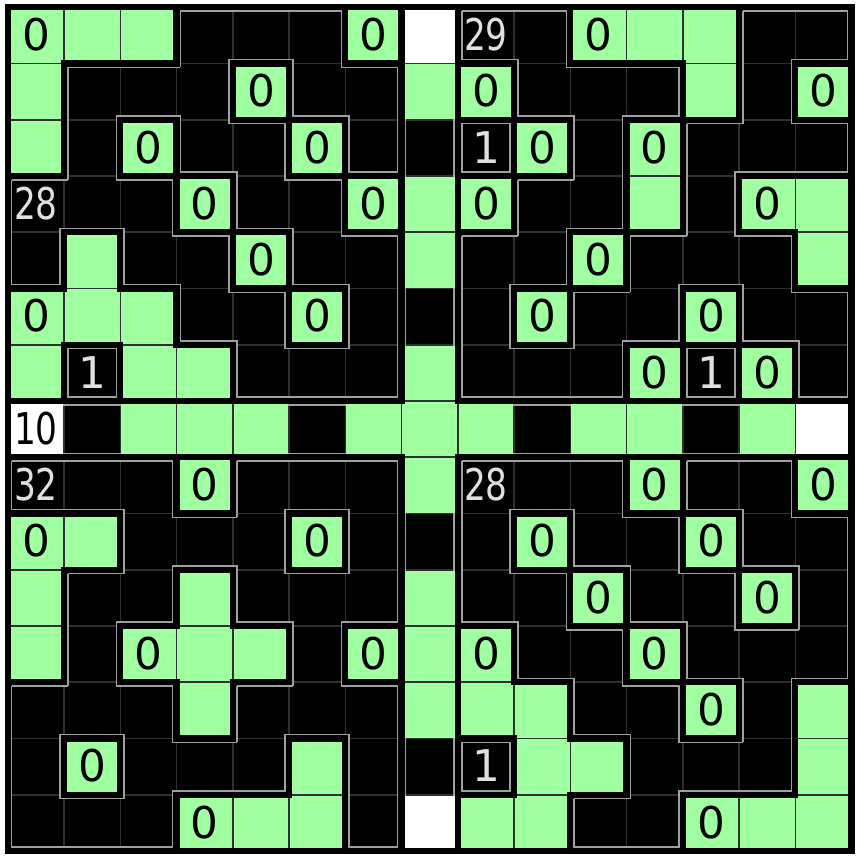} \puzpdfs{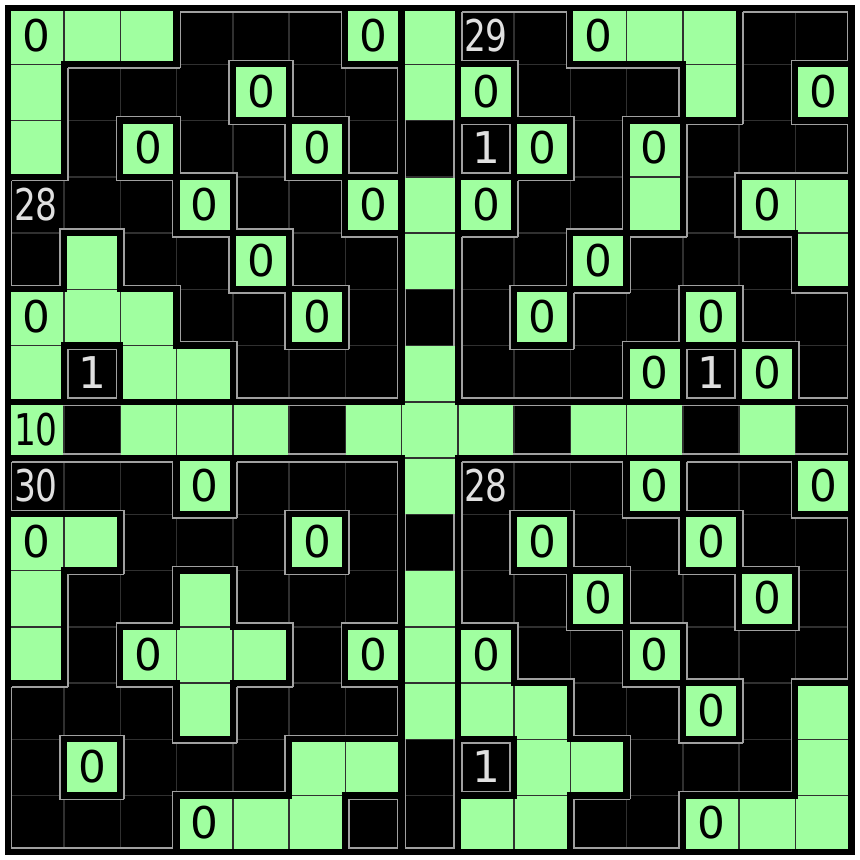}\end{center}

\subsubsection{Paintarea}

Paintarea is a shading genre.
The grid is divided into regions, each of which must be either entirely
shaded or entirely unshaded.
The shaded cells must be connected, and no $2\times2$ square may be entirely shaded or entirely unshaded.
Numbered cells indicate the number of orthogonally adjacent shaded cells.

We construct a required-edge undirected asymmetric T-metacell to show ASP-completeness of Paintarea without any numbered cells (Corollary~\ref{cor:tundirforced}).

The base for our T-metacell is this $15 \times 15$ outer frame,
which ensures that shaded cells can connect only in the middle of their shared edges.

\puzpng{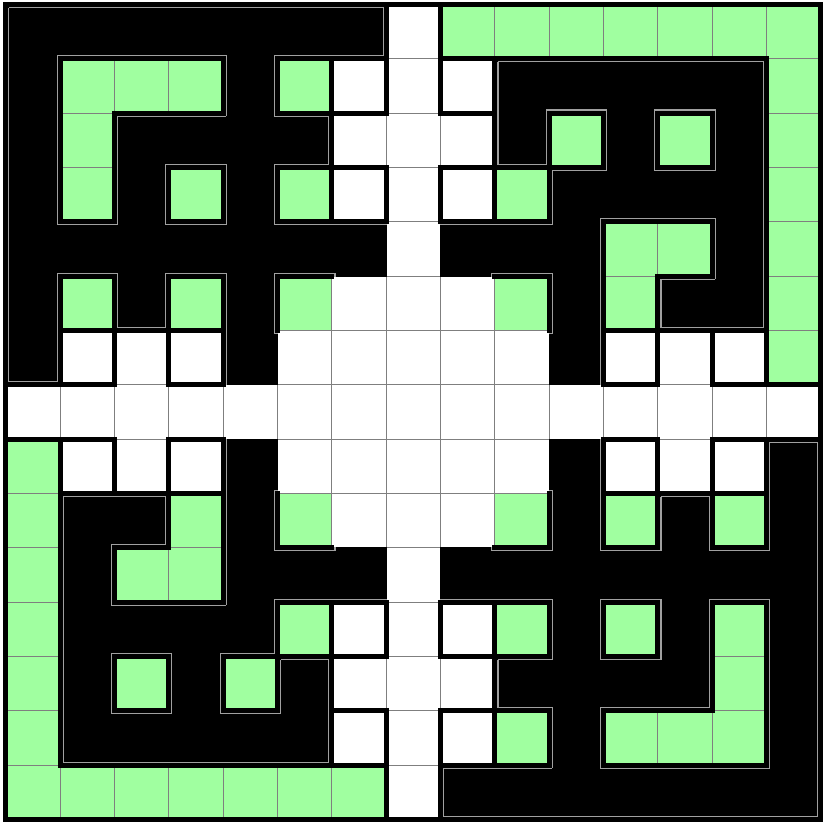}

The center of the frame is then filled with the following $7 \times 7$ core, which can be individually rotated and reflected (while the outer frame is fixed) to produce a T-metacell with the desired orientation.

\puzpng{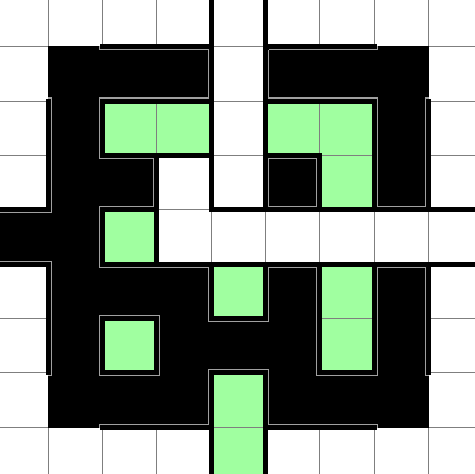}

We also need to break the loop into a path, which is accomplished by using this tile for the top-right corner:

\puzpng{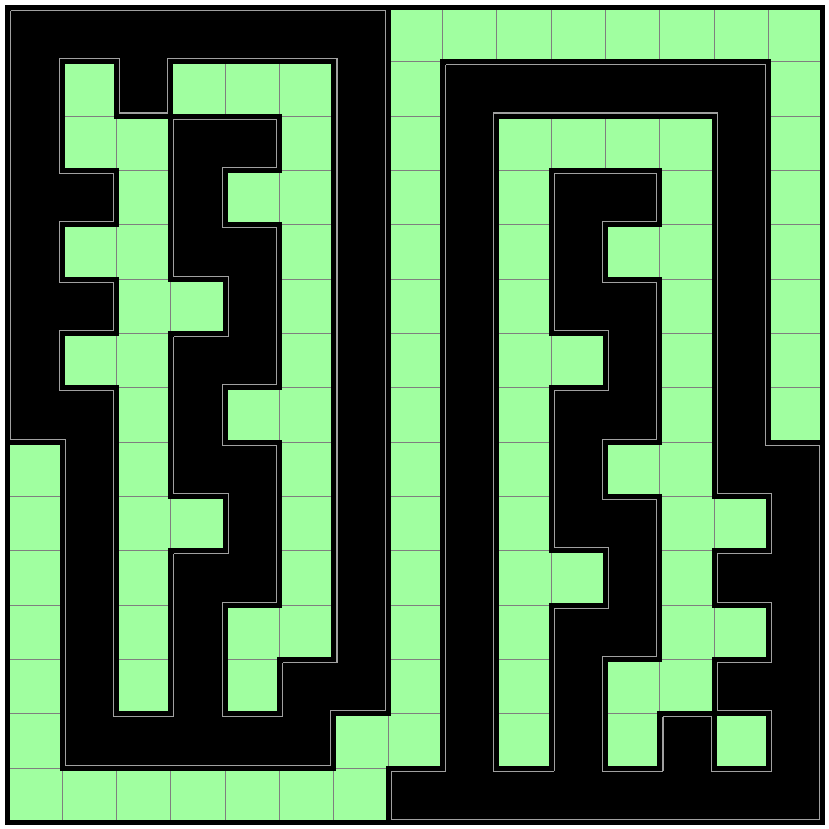}

\subsection{Open ASP-Completeness Questions}

Some puzzle genres were proved NP-complete by Tang,
but we have not yet found parsimonious adaptations of the corresponding T-metacells.
These genres are
Angle Loop,
Double Back,
Scrin,
Icebarn,
and Icelom 2.

\section*{Acknowledgments}

This paper was initiated during open problem solving in the MIT class
on Algorithmic Lower Bounds: Fun with Hardness Proofs (6.5440)
taught by Erik Demaine in Fall 2023.
We thank the other participants of that class for helpful discussions
and providing an inspiring atmosphere.

Some figures were drawn using SVG Tiler [\url{https://github.com/edemaine/svgtiler}],
and some were drawn using puzz.link \cite{puzzlink}.

\bibliography{gridgraphs}
\bibliographystyle{alpha}

\end{document}